\documentclass[sigconf, screen, 9pt]{acmart}

\AtBeginDocument{%
  }

\copyrightyear{2024}
\acmYear{2024}
\setcopyright{acmlicensed}\acmConference[LICS '24]{39th Annual ACM/IEEE Symposium on Logic in Computer Science}{July 8--11, 2024}{Tallinn, Estonia}
\acmBooktitle{39th Annual ACM/IEEE Symposium on Logic in Computer Science (LICS '24), July 8--11, 2024, Tallinn, Estonia}
\acmDOI{10.1145/3661814.3662086}
\acmISBN{979-8-4007-0660-8/24/07}

\usepackage{mathtools}
\usepackage{bm}
\usepackage{cleveref}
\usepackage{tikz}
\usetikzlibrary{petri, positioning, fit, backgrounds, patterns, automata}
\usepackage{macros}

\newenvironment{claimproof}[1]{\par\noindent\underline{Proof:}\space#1}{\hfill $\blacksquare$}
\newtheorem{remark}{Remark}

\begin{document}

\title{Soundness of reset workflow nets}

\author{Michael Blondin}
\authornote{Supported by a Discovery Grant from the Natural Sciences and Engineering Research Council of Canada (NSERC), and by the Fonds de recherche du Qu\'{e}bec -- Nature et technologies (FRQNT).}
\email{michael.blondin@usherbrooke.ca}
\orcid{0000-0003-2914-2734}
\affiliation{%
  \institution{Universit\'{e} de Sherbrooke}
  \city{Sherbrooke}
  \state{Qu\'{e}bec}
  \country{Canada}
}

\author{Alain Finkel}
\email{alain.finkel@ens-paris-saclay.fr}
\orcid{0000-0003-2482-6141}
\affiliation{%
  \institution{Universit\'{e} Paris-Saclay, CNRS, ENS Paris-Saclay, IUF}
  \streetaddress{Laboratoire M\'{e}thodes Formelles}
  \city{Gif-sur-Yvette}
  \country{France}
}

\author{Piotr Hofman}
\authornote{Supported by an ERC Starting grant INFSYS, agreement no. 950398.}
\email{piotr.hofman@uw.edu.pl}
\orcid{0000-0001-9866-3723}
\affiliation{%
  \institution{University of Warsaw}
  \city{Warsaw}
  \country{Poland}
}

\author{Filip Mazowiecki}
\authornote{Supported by Polish National Science Centre SONATA BIS-12 grant number 2022/46/E/ST6/00230.}
\email{f.mazowiecki@mimuw.edu.pl}
\orcid{0000-0002-4535-6508}
\affiliation{%
  \institution{University of Warsaw}
  \city{Warsaw}
  \country{Poland}
}

\author{Philip Offtermatt}
\email{philip.offtermatt@informal.systems}
\orcid{0000-0001-8477-2849}
\affiliation{%
  \institution{Informal Systems}
  \city{Lausanne}
  \country{Switzerland}
}


\begin{abstract}
  Workflow nets are a well-established variant of Petri nets for the
  modeling of process activities such as business processes. The
  standard correctness notion of workflow nets is soundness, which
  comes in several variants. Their decidability was shown decades ago,
  but their complexity was only identified recently. In this work, we
  are primarily interested in two popular variants: $1$-soundness and
  generalised soundness.
  
  Workflow nets have been extended with resets to model workflows that
  can, e.g., cancel actions. It has been known for a while that, for
  this extension, all variants of soundness, except possibly
  generalised soundness, are undecidable.
  
  We complete the picture by showing that generalised soundness is
  also undecidable for reset workflow nets. We then blur this
  undecidability landscape by identifying a
  property, coined ``$1$-in-between soundness'', which lies between
  $1$-soundness and generalised soundness. It reveals an unusual
  non-monotonic complexity behaviour: a decidable soundness property is
  in between two undecidable ones. This can be valuable in the
  algorithmic analysis of reset workflow nets, as our procedure
  yields an output of the form ``$1$-sound'' or ``not generalised
  sound'' which is always correct.
\end{abstract}




\maketitle

\section{Introduction}
\emph{Workflow nets} are a well-established formalism for the modeling of process activities such as business processes~\cite{AL97}. For
example, they can be used as the formal representation of workflow
procedures in business process management systems (see
e.g.~\cite[Section~4]{van1998application} and~\cite[Section~3]{AL97}
for details on modeling procedures). Workflow nets enable the
algorithmic formal analysis of their behaviour. This is relevant,
e.g., for organizations that seek to manage complex business
processes. For example, according to a survey~\cite{AHHSVVW11}, over
20\% instances from the SAP reference model have been detected to be
flawed (due to deadlocks, livelocks, etc.)

A workflow net $\W$ is essentially a Petri net --- a prominent model
of concurrency --- that satisfies extra properties. In particular, $\W$ has two
designated places $i$ and $f$ respectively called \emph{initial} and
\emph{final}. Initially, $\W$ starts with $k$ tokens in place $i$ that
can evolve according to the transitions of $\W$, which can consume and
create tokens. Informally, a token that reaches $f$ indicates
that some activity has been completed.

\subsection{Soundness}

The standard correctness notion of workflow nets is soundness. Various
definitions have been considered in the literature. Most prominently,
\emph{$k$-soundness} requires that, starting from $k$ tokens in the
initial place, no matter what transitions are taken, it is always
possible to complete properly, i.e. to end up with $k$ tokens in the
final place, and no token elsewhere. \emph{Generalised soundness}
requires a net to be $k$-sound for all $k > 0$, while \emph{structural
soundness} requires $k$-soundness for some $k > 0$. \emph{Classical
soundness} requires $1$-soundness and each transition to be fireable
in at least one execution.

The decidability of soundness was established some two decades
ago~\cite{AL97,HeeSV04,TM05,vOS06} (see~\cite{AHHSVVW11} for a
survey). The underlying algorithms relied on Petri net reachability,
which was recently shown to be non-primitive
recursive~\cite{LS19,Ler21,CO21}. Until recently, no better bound was
known. At LICS'22~\cite{BlondinMO22}, the computational complexity of
all variants of soundness was established: generalised soundness is
PSPACE-complete, while the other variants are EXPSPACE-complete.

One shortcoming of workflow nets is that they lack useful features like cancellation. As mentioned
in~\cite{vvtSVVW08}, ``[m]any practical languages have a cancelation
feature, e.g., Staffware has a withdraw construct, YAWL has a
cancelation region, BPMN has cancel, compensate, and error events,
etc.'' Thus, workflow nets have been extended with \emph{resets} that
instantly remove all tokens from a region of the net upon taking a
transition (e.g.\ see~\cite{AHHSVVW09}). Around fifteen years ago, it
was shown that, for reset workflow nets,
\begin{quote}
  \emph{[... of] the many soundness notions described in [the] literature
    only generalised soundness may be decidable (this is still an open
    problem). All other notions are shown to be
    undecidable.}~\cite{vvtSVVW08}
\end{quote}

So, while all studied variants such as $1$-soundness, $k$-soundness,
structural soundness and classical soundness are
undecidable for reset workflow nets, the decidability of generalised
soundness has remained open.

The recent results of~\cite{BlondinMO22} shows that in the absence of resets, generalised
soundness is computationally easier than other soundness variants.
This was an indication that generalised
soundness could remain decidable in the presence of resets.

\subsection{Our contribution}

In this work, we first show that such a conjecture is false:
generalised soundness for reset workflow nets is undecidable.

The undecidability of the various types of soundness in reset workflow
nets indicates the necessity for a new approach. Consequently, we
embarked on the journey of exploring potential methods for
``approximating'' soundness. We propose a precise definition of an
acceptable approximation.

More precisely, we say that a property $\mathcal{P}$ of reset
workflow nets is a \emph{$1$-in-between soundness} property if it
meets the following criteria: all generalised sound reset workflow
nets satisfy $\mathcal{P}$, and every reset workflow net satisfying
$\mathcal{P}$ is $1$-sound. Remarkably, we have discovered such a property,
denoted as $\mathcal{P}_1$, which is
decidable. Thus, we have
\[
\text{Generalised sound nets} \subseteq 
\mathcal{P}_1 \text{ nets} \subseteq
\text{$1$-sound nets},
\]
where generalised soundness and $1$-soundness are
undecidable, but $\mathcal{P}_1$ is decidable.

This reveals a non-monotonic complexity behaviour. By this, we mean
three properties $A$, $B$ and $C$ such that
\[
A~\text{(hard)} \subseteq 
B~\text{(easy)} \subseteq
C~\text{(hard)}.
\]
This phenomenon can be regarded as unnatural and is rarely actively
exploited. We present two other such examples.

First, consider the setting where we are given two graphs, each with a
designated initial node and with directed edges labelled from a common
finite alphabet (so, finite automata whose states are all
accepting). These inclusions hold:
\[
  \text{Isomorphism} \subseteq
  \text{Bisimilarity} \subseteq
  \text{Trace(-language) equivalence},
\]
The first property is known to be checkable in quasi-polynomial time (a
survey on current advances can be found in~\cite{10.1145/3372123}),
the second is checkable in polynomial
time~\cite{DBLP:journals/siamcomp/PaigeT87}, and the third one is 
well known to be PSPACE-complete.

Now, consider the setting where we are given two one-counter nets
(OCN), where an OCN is an automaton with a single counter over the
naturals that can be incremented and decremented (but not
zero-tested). These inclusions hold:
\begin{alignat*}{3}
  &\text{Weak bisimulation~\cite{DBLP:conf/icalp/Mayr03}}
  \subseteq {} \\
  &\text{Weak simulation in both
    directions~\cite{DBLP:journals/corr/HofmanLMT16}} \subseteq {} \\
  &\text{Language equivalence~\cite{DBLP:phd/ethos/Valiant73}},
\end{alignat*}
The first and the last properties are undecidable, while the middle
one is PSPACE-complete.

Our property $\mathcal{P}_1$ can be valuable in the analysis
of reset workflow nets. For example, we provide an algorithm that,
on any reset workflow net which is not $1$-sound, classifies it as not
generalised sound, and on any generalised sound reset workflow net,
guarantees at least $1$-soundness. For reset workflow nets that are
$1$-sound but not generalised sound, the algorithm provides the
correct description of either not being generalised sound or being
$1$-sound. It is worth noting that all these answers accurately
characterize the given reset workflow net.

However, the computational complexity of verifying $\mathcal{P}_1$ is
non-primitive recursive in the worst case, which hinders immediate
applications. Nevertheless, we cannot rule out the existence of
similar predicates with better computational complexity, or of
implementations performing faster on real-world instances. This area
requires further investigation.

The definition of $\mathcal{P}_1$ is exceedingly technical, making it
challenging to convey basic intuitions concisely. We actually define a decidable family of properties
$\{\mathcal{P}_k\}_{k > 0}$.
In our final result, we prove a connection with another notion
of soundness found in the literature, namely \emph{up-to-$k$
  soundness}~\cite{phdthesisToorn}. Specifically, we demonstrate that
for every reset workflow net $\W$, there is a computable but
substantially large value $K$ such that, for every $k > K$, $\W$ is
up-to-$k$ sound if and only if it satisfies the property~$\mathcal{P}_k$.

\subsection{Further related work}

A significant portion of the work that relates to reset workflow nets
consists of results on reset Petri nets, in particular, on
\begin{itemize}
 \item The theory of well-structured transition systems and results
   for the general class of reset Petri
   nets~\cite{FINKEL20041,DBLP:conf/icalp/DufourdFS98};

 \item Results on restricted subclasses, such as those with a limited
   number of places that can be
   reset~\cite{DBLP:conf/fsttcs/FinkelLS18}, acyclic reset Petri nets
   and workflow nets~\cite{CHMS23}, or Petri nets with hierarchical
   reset arcs~\cite{DBLP:conf/concur/AkshayCDJS17};

 \item More practically oriented results on relaxations of the
   reachability relation like the integer
   relaxation~\cite{DBLP:conf/rp/HaaseH14}.
\end{itemize}

An important branch of soundness-related research includes work on
reduction rules that preserve (un)soundness while reducing the size
and structure of the workflow net. Rules specific to reset workflow
nets are explored in~\cite{WVAHE09,WVAHE09b}.

Another line of research lies in the field of process discovery,
e.g.\ the discovery of cancellation regions within process mining
techniques~\cite{KL14}.

\subsection{Paper organization}

The paper is organized as follows. In \Cref{sec:preliminaries}, we
introduce preliminary definitions such as Petri nets, workflow nets
and notions of soundness. In \Cref{sec:undecidability}, we establish
the undecidability of generalised soundness for reset workflow
nets. \Cref{sec:between} shows the existence of the $k$-in-between
soundness property $\mathcal{P}_k$, and relates it with another notion
of soundness. In particular, \Cref{ssec:redundancy} introduces the
intermediate notion of ``nonredundancy'', and \Cref{ssec:skeleton}
introduces the intermediate notion of ``skeleton'' workflow nets.

\section{Preliminaries}
\label{sec:preliminaries}
Let $\N \defeq \{0, 1, \ldots\}$ and $\Npos \defeq \N \setminus
\{0\}$. For every $a, b \in \N$ such that $a \leq b$, we define
$[a..b] \defeq \{a, a + 1, \ldots, b\}$. Given a set $X$, we denote
its cardinality by $|X|$.

Let $\m, \m' \colon P \to \N$ where $P$ is a finite set. We see $\m$
and $\m'$ as both unordered vectors and multisets. For example, we
have $\m = \vec{0}$ if $\m(p) = 0$ for every $p \in P$, and we have
$\m' = \{p \colon 1, q \colon 3\}$ if $\m'(p) = 1$, $\m'(q) = 3$ and
$\m'(r) = 0$ for every $r \in P \setminus \{p, q\}$. We write $\m \leq
\m'$ iff $\m(p) \leq \m'(p)$ for every $p \in P$. We write $\m < \m'$
if $\m \leq \m'$ and $\m \neq \m'$. We define $\m + \m'$ as the
mapping satisfying $(\m + \m')(p) = \m(p) + \m'(p)$ for every $p \in
P$. The mapping $\m - \m'$ is defined similarly, provided that $\m
\geq \m'$. For every $Q \subseteq P$, we define $\m(Q) \defeq \sum_{q
  \in Q} \m(q)$.

\subsection{Ackermannian complexity}

In the upcoming sections, we prove theorems about the existence of some numbers bounded by functions on the size of the input. To do this rigorously, we should refer to these individual bounds and precisely track dependencies between them. However, some of the bounds we use are based on the Ackermann function (which is non-primitive recursive), and tracking them precisely is tedious. Moreover, keeping all the constants in mind would be troublesome for the reader. That is why we choose to simply state that some number is of ``Ackermannian size'' or''Ackermannianly bounded''. In our reasoning, we use the following rules:
\begin{itemize}
\item The maximum, sum and product of Ackermannianly bounded constants yields Ackermannianly bounded constants;

%
\item The application of a function within the Ackermannian class of functions to a bound of Ackermannian size yields an Ackermannianly bounded constant.
\end{itemize}
Since we will make a constant number of such manipulations, the
Ackermannian bounds will be preserved.

For an introduction to high computational complexity classes (from the
fast-growing hierarchy), we refer the reader
to~\cite{lob1970hierarchies,DBLP:journals/toct/Schmitz16}.

\subsection{Petri nets}

A \emph{reset Petri net} $\PN$ is a tuple $(P, T, F, R)$ where
\begin{itemize}
\item $P$ is a finite set of elements called \emph{places},

\item $T$ is a finite set, disjoint from $P$, of elements called
  \emph{transitions},

\item $F \subseteq (P \times T) \cup (T \times P)$ is a set of
  elements called \emph{arcs},

\item $R \subseteq (P \times T)$ is a set of elements called
  \emph{reset arcs}.
\end{itemize}
A \emph{(standard) Petri net} is a reset Petri net with no reset arc
($R = \emptyset$).

\begin{example}
  \Cref{fig:example:pn} depicts a reset Petri net $\PN = (P, T, F, R)$ where $P
  = \{p_1, p_2, p_3, p_4\}$ (circles), $T = \{t_1, t_2, t_3\}$
  (boxes), each arc $(u, v) \in F$ is depicted by a directed edge, and
  the only reset arc is depicted by a dotted directed edge. \qed
\end{example}


Given a transition $t \in T$, we define $\pre{t} \defeq \{p \in P :
(p, t) \in F\}$ and $\post{t} \defeq \{p \in P : (t, p) \in F\}$. Both
of these sets will often be interpreted implicitly as mappings from
$P$ to $\{0, 1\}$. For example, we can write either ``$p \in
\post{t}$'' or ``$\post{t}(p) = 1$''. In \Cref{fig:example:pn}, we
have, e.g., $\pre{t_1} = \{p_1\}$, $\post{t_1} = \{p_2, p_3\}$ and
$\pre{t_3} = \{p_2\}$. Given a place $p \in P$, we define $\pre{p}
\defeq \{t \in T : (t, p) \in F\}$ and $\post{p} \defeq \{t \in T :
(p, t) \in F\}$.

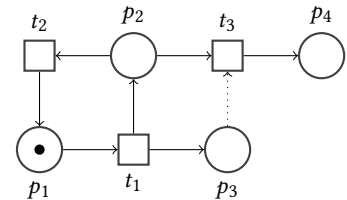
\begin{figure}[!b]
  \begin{center}
    \begin{center}
  \begin{tikzpicture}[auto, node distance=1.25cm]
    \tikzstyle{ctrans} = [transition]
    
    \node[place, label=below:$p_1$, tokens=1] (p1)  {};
    \node[transition, right of=p1, label=below:$t_1$] (t1) {};
    \node[place, above of=t1, label=above:$p_2$] (p2)  {};
    \node[transition, left of=p2, label=above:$t_2$] (t2) {};
    \node[place, right of=t1, label=below:$p_3$] (p3)  {};
    \node[transition, above of=p3, label=above:$t_3$] (t3) {};
    \node[place, right of=t3, label=above:$p_4$] (p4)  {};

    \path[->]
    (p1) edge node {} (t1)
    (t1) edge node {} (p2)
    (t1) edge node {} (p3)

    (p2) edge node {} (t2)
    (t2) edge node {} (p1)

    (p2) edge node {} (t3)
    (t3) edge node {} (p4)
    ;

    \path[->, dotted]
    (p3) edge node {} (t3)
    ;
  \end{tikzpicture}
\end{center}
  \end{center}
  \caption{Example of a reset Petri net where circles are places,
    boxes are transitions, solid edges are arcs, and dotted edges are
    reset arcs. The small filled circle depicts a
    token.}\label{fig:example:pn}
\end{figure}

A \emph{marking} is a mapping $\m \colon P \to \N$ that indicates the number
of \emph{tokens} in each place. Given a marking $\m$ and a transition $t$, we
let $\Rest{t}{\m}$ denote the marking obtained from $\m$ by emptying
all places that are reset by $t$. Formally, $\Rest{t}{\m} \defeq 0$ if
$(p, t) \in R$, and $\m(p)$ otherwise.

We say that a transition $t \in T$ is \emph{enabled} in marking $\m$
if $\m \geq \pre{t}$, i.e.\ if each place of $\pre{t}$ contains at
least one token. If $t$ is enabled in $\m$, then it can be
\emph{fired}. In words, upon firing $t$, a token is consumed from each
place of $\pre{t}$; then, all places of $\{p : (p, t) \in R\}$ are
emptied; and, finally, a token is produced in each place of
$\post{t}$. More formally, firing $t$ leads to the marking $\m'$
defined as follows, for every $p \in P$:
\[
\m'(p)
=
\begin{cases}
  \m(p) - \pre{t}(p) + \post{t}(p) & \text{if } (p, t) \notin R, \\
  \post{t}(p) & \text{otherwise}.
\end{cases}
\]
Equivalently, and more succinctly, $\m' = \Rest{t}{\m - \pre{t}} +
\post{t}$.

We write $\m \trans{t} \m'$ whenever $t$ is enabled in $\m$ and firing
$t$ from $\m$ leads to $\m'$. We write $\m \trans{} \m'$ if $\m
\trans{t} \m'$ holds for some $t \in T$. We write ${\trans{*}}$ to
denote the reflexive-transitive closure of ${\trans{}}$. Given a
subset $X$ of markings, we write $\m \trans{*} X$ to denote that there
exists $\m' \in X$ such that $\m \trans{*} \m'$.

Given a sequence of transitions $\pi = t_1 t_2 \cdots t_n$ and a
transition $s$, we define $|\pi| \defeq n$ and $|\pi|_s = |\{i \in
[1..n] : t_i = s\}|$.

\begin{example}
  Reconsider the reset Petri net of \Cref{fig:example:pn}. For the
  sake of brevity, let us write a marking $\{p_1 \colon a, p_2 \colon
  b, p_3 \colon c, p_4 \colon d\}$ as $(a, b, c, d)$. From marking
  $(1, 0, 0, 0)$, we can only fire $t_1$. Firing $t_1$ leads to $(0,
  1, 1, 0)$. From the latter, we can fire either $t_2$ or
  $t_3$. Firing $t_2$ leads to $(1, 0, 1, 0)$. From there, we can only
  fire $t_1$, which leads to $(0, 1, 2, 0)$. From the latter, we can
  fire either $t_2$ or $t_3$. Firing $t_3$ leads to $(0, 0, 0, 1)$,
  from which no transition is enabled. The described sequence can be
  written as
  \begin{multline*}
    (1, 0, 0, 0)
    \trans{t_1}
    (0, 1, 1, 0)
    \trans{t_2}
    (1, 0, 1, 0) \\
    \trans{t_1}
    (0, 1, 2, 0)
    \trans{t_3}
    (0, 0, 0, 1).
  \end{multline*}
  Thus, we have $(1, 0, 0, 0) \trans{t_1\, t_2\, t_1\, t_3} (0, 0, 0,
  1)$, or more succinctly $(1, 0, 0, 0) \trans{*} (0, 0, 0, 1)$. \qed
\end{example}

\subsection{Upward and downward closed sets}

The set of markings $\N^P$ is partially ordered by ${\leq}$. The
latter is a well-quasi-order, which means that it neither contains
infinite antichains nor infinite (strictly) decreasing sequences.

A set of markings $X \subseteq \N^P$ is \emph{upward closed} if for
every $\m \in X$ and every other marking $\m' \in \N^P$, it is the
case that $\m \leq \m'$ implies $\m' \in X$. The \emph{upward closure}
of $Y \subseteq \N^P$ is defined as the smallest upward closed set,
denoted $\upclosure{Y}$, such that $Y \subseteq \upclosure{Y}$.

Similarly, a subset $X \subseteq \N^P$ is \emph{downward closed} if
for every $\m \in X$ and every other marking $\m' \in \N^P$, we have
$\m' \leq \m \implies \m' \in X$. The \emph{downward closure} of
$Y \subseteq \N^P$, denoted $\downclosure{Y}$, is the smallest
downward closed set such that $Y \subseteq \downclosure{Y}$.

We extend upward and downward closures to individual markings, i.e.\
$\upclosure{\m} \defeq \upclosure{\{\m\}}$ and
$\downclosure{\m} \defeq \downclosure{\{\m\}}$.

It is worth noting that every upward closed set is characterized by
the set of its minimal elements. Since such a set is an antichain, it
must be finite, resulting in a finite representation for each upward
closed set. Similarly, every downward closed set can be finitely
represented by a finite representation of its complement, which is an
upward closed set.

Given a reset Petri net $\PN$, we say that from a marking $\m$, it is
possible to \emph{cover} a marking $\m'$ if there exists a marking
$\m'' \geq \m'$ such that $\m \trans{*} \m''$ holds in $\PN$. Let
$C_{\m'}$ be the set of all markings $\m$ from which it is possible to
cover a marking $\m'$. Observe that $C_{\m'}$ is upward closed. An
important result is that, for any given marking $\m'$, it is possible
to compute the minimal elements of the set $C_{\m'}$. This computation
can be achieved using the so-called ``backward coverability
algorithm''~\cite{DBLP:journals/tcs/FinkelS01}.

Additionally, it is crucial to note the following lemma.
\begin{lemma}[\cite{FigueiraFSS11}]\label{lem:boundOnCoverability}
If it is possible to cover
$\m'$ from $\m$, then it can be done with a run whose length is Ackermannianly bounded in the size of the Petri net and $\m'$.
\end{lemma}
Throughout the paper we will often use \Cref{lem:boundOnCoverability} without referring to it.
We will even use a stronger property that the set of minimal elements
of $C_{\m'}$ can be computed in Ackermannian time. This follows from~\cite{LazicS21}, where it is proved that the backward coverability algorithm works in Ackermannian time.


\subsection{Workflow nets and soundness}

A \emph{reset workflow net} is a tuple $(P, T, F, R, i, f)$ where
\begin{itemize}
\item $\PN = (P, T, F, R)$ is a reset Petri net,

\item $i \in P$ is a place called \emph{initial} that satisfies $\pre{i}
  = \emptyset$,

\item $f \in P$ is a place called \emph{final} that satisfies $\post{f}
  = \emptyset$, and $(f, t) \notin R$ for all $t \in T$ (no transition
  resets $f$),

\item each element of $P \cup T$ is on some path from $i$ to $f$ in
  the underlying graph of $\PN$ without considering reset arcs,
  i.e.\ in the graph $G \defeq (V, E)$ with vertices $V \defeq P \cup
  T$ and directed edges $E \defeq \{(u, v) \in V \times V : (u, v) \in
  F\}$.
\end{itemize}
A \emph{(standard) workflow net} is a reset workflow net with no reset
arc, i.e.\ with $R = \emptyset$. For example, \Cref{fig:example}
depicts a reset workflow net.

Given $k \in \Npos$, we say that a reset workflow net $\W$
is \emph{$k$-sound} if for every marking $\m$ such that $\{i \colon
k\} \trans{*} \m$, it is the case that $\m \trans{*} \{f \colon
k\}$. In other words, $\W$ is $k$-sound if, from $k$ tokens in the
initial place, no matter what is fired, it is always possible to end
up with $k$ tokens in the final place (and no token elsewhere). We
say that a reset workflow net $\W$ is \emph{generalised sound} if it
is $k$-sound for all $k \in
\Npos$.

A marking $\m$ is a \emph{witness of $k$-unsoundness} if $\{i \colon
k\} \trans{*} \m$ and $\m \not\trans{*} \{f \colon k\}$. A marking
$\m$ is a \emph{witness of unsoundness} if it is a witness of
$k$-unsoundness for some $k \in \Npos$.

\begin{example}
  Reconsider the reset Petri net of \Cref{fig:example:pn}. Taking $i
  \defeq p_1$ and $f \defeq p_4$ does not yield a reset workflow net
  for the two following reasons:
  \begin{itemize}
  \item we have: $\pre{p_1} = \{t_2\} \neq \emptyset$, and

  \item there is no path from $p_3$ to $p_4$ (along solid edges).
  \end{itemize}

  \Cref{fig:example} depicts a reset workflow net $\W$. The set of
  markings reachable from $\{i \colon 1\}$ in $\W$ is depicted in
  \Cref{fig:reach}. It is readily seen that any reachable marking $\m$
  can reach $\{f \colon 1\}$. Thus, $\W$ is $1$-sound. However, $\W$
  is not $2$-sound and hence not generalised sound. Indeed, we have
  \begin{multline*}
    \{i \colon 2\}
    \trans{ss}
    \{p_1 \colon 2, p_2 \colon 2\}
    \trans{t_2} {} \\
    \{p_1 \colon 2, p_2 \colon 1, q_2 \colon 1, q_3 \colon 1\}
    \trans{u_2}
    \{f \colon 1\}.
  \end{multline*}
  As $\{f \colon 1\}$ cannot reach any other marking, it cannot reach
  $\{f \colon 2\}$ as required. So, $\{f \colon 1\}$ is a witness of
  $2$-unsoundness, and hence of unsoundness. \qed
\end{example}

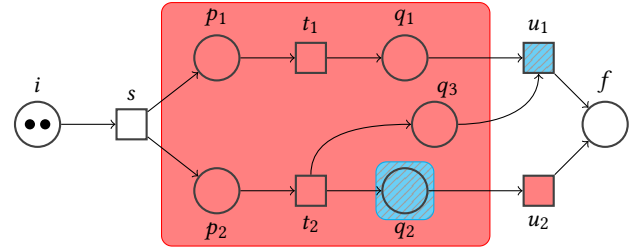
\begin{figure}
  \begin{center}
    \begin{center}
  \begin{tikzpicture}[auto, node distance=1.25cm]
    \tikzstyle{ctrans} = [transition]
    
    \node[place, label=above:$i$, tokens=2] (i)  {};

    \node[transition, right of=i, label=above:$s$] (t0) {};

    \node[place, above right of=t0, xshift=0.25cm, label=above:$p_1$] (p1) {};
    \node[place, below right of=t0, xshift=0.25cm, label=below:$p_2$] (p2) {};

    \node[transition, right of=p1, label=above:$t_1$] (t1) {};
    \node[transition, right of=p2, label=below:$t_2$] (t2) {};

    \node[place, right of=t1, label=above:$q_1$] (q1) {};
    \node[place, right of=t2, label=below:$q_2$] (q2) {};
    \node[place, right=3.5cm of t0, label={[xshift=5pt, yshift=-2pt]$q_3$}] (q3) {};

    \node[
      ctrans, right=1.25cm of q1, label=above:$u_1$,
      pattern=north east lines, pattern color=gray!75,
      preaction={fill=colB!50}
    ] (u1) {};

    \node[
      ctrans, right=1.25cm of q2, label=below:$u_2$,
      fill=colA!50
    ] (u2) {};

    \node[place, label=above:$f$, above right of=u2] (f)  {};

    \path[->]
    (i) edge node {} (t0)

    (t0) edge node {} (p1)
    (t0) edge node {} (p2)

    (p1) edge node {} (t1)
    (t1) edge node {} (q1)

    (p2) edge node {} (t2)
    (t2) edge node {} (q2)
    (t2) edge[out=90, in=180] node {} (q3)

    (q1) edge node {} (u1)
    (q3) edge[out=0, in=-90] node {} (u1)
    (u1) edge node {} (f)

    (q2) edge node {} (u2)
    (u2) edge node {} (f)
    ;
    
    \begin{pgfonlayer}{background}
      \node[fit=(p1)(p2)(q1)(q2)(q3),
        draw=colA, rounded corners, fill=colA!50, inner sep=12pt] {};

      \node[fit=(q2), pattern=north east lines, pattern color=gray!75,
        draw=colB, rounded corners, preaction={fill=colB!50},
        inner sep=2pt] {};
    \end{pgfonlayer}
  \end{tikzpicture}
\end{center}
  \end{center}
  \caption{Example of a reset workflow net. Reset arcs are depicted
    implicitly by colored patterns, rather than explicitly by dotted
    directed edges. In words, transition $u_1$ resets place $q_2$, and
    transition $u_2$ resets all places from $\{p_1, p_2, q_1, q_2,
    q_3\}$. Formally, $R = \{(q_2, u_1)\} \cup \{(r, u_2) : r \in
    \{p_1, p_2, q_1, q_2, q_3\}\}$. The two filled circles within
    place $i$ represent two tokens.}\label{fig:example}
\end{figure}

\begin{figure}[h!]
  \begin{center}
    \hspace*{25pt}
    \begin{tikzpicture}[auto]
  \node[] (a) {$\{i \colon 1\}$};

  \node[below=15pt of a] (b) {$\{p_1 \colon 1, p_2 \colon 1\}$};

  \node[below=5pt of b, xshift=-50pt] (c) {
    $\{q_1 \colon 1, p_2 \colon 1\}$
  };

  \node[below=5pt of b, xshift=50pt] (d) {
    $\{p_1 \colon 1, q_2 \colon 1, q_3 \colon 1\}$
  };

  \node[below=5pt of c, xshift=50pt] (e) {
    $\{q_1 \colon 1, q_2 \colon 1, q_3 \colon 1\}$
  };

  \node[below=15pt of e] (f) {$\{f \colon 1\}$};

  \path[->, black!60]
  (a) edge node[swap] {$s$} (b)

  (b) edge[out=180, in=90] node[swap] {$t_1$} (c)
  (b) edge[out=0,   in=90] node[]     {$t_2$} (d)

  (c) edge[out=-90, in=180] node[swap] {$t_2$} (e)
  (d) edge[out=-90, in=0]   node[]     {$t_1$} (e)

  (e) edge node[swap] {$u_1, u_2$} (f)

  (d) edge[out=0, in=0, looseness=1] node[] {$u_2$} (f)
  ;
\end{tikzpicture}
  \end{center}
  \caption{The set of markings reachable from $\{i \colon 1\}$ in the
    reset workflow net of \Cref{fig:example}. Each edge $\m \trans{t}
    \m'$ indicates that firing transition $t$ in marking $\m$ leads to
    marking $\m'$.}\label{fig:reach}
\end{figure}
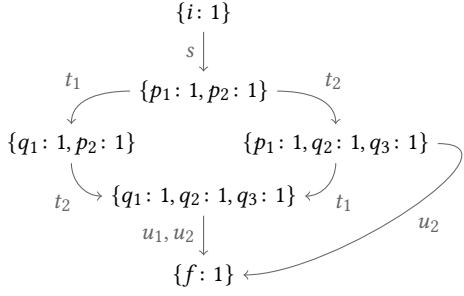

\subsection{Subnets}\label{ssec:subnets}

Let $\W = (P, T, F, R, i, f)$ be a reset workflow net, let $Q
\subseteq P$ and let $S \subseteq T$. The reset Petri net
\emph{obtained from $\W$ by removing places $Q$ and transitions $S$}
is the reset Petri net $(P, T, F, R)$ from which we remove $Q$, $S$
and any remaining isolated node (i.e.\ with no incoming and outgoing
arc.) More formally, it is the reset Petri net $\PN \defeq (P', T',
F', R')$ where
\begin{align*}
  P' &\defeq \{p \in P \setminus Q : (\pre{p} \cup \post{p})
  \not\subseteq S\} \\
  T' &\defeq \{t \in T \setminus S : (\pre{t} \cup \post{t})
  \not\subseteq Q\}, \\
  F' &\defeq \{(p, t) \in F : p \in P', t \in T'\} \cup \{(t, p) \in F
  : p \in P', t \in T'\}, \\
  R' &\defeq \{(p, t) \in R : p \in P', t \in T'\}.
\end{align*}

By definition, $\PN$ has no isolated node, i.e.\ with no incoming or
outgoing arc. This holds even if we only remove places ($S =
\emptyset$) or only remove transitions ($Q = \emptyset$). Indeed,
since $\W$ is a reset workflow net, each place $p \in P$ and each
transition $t \in T$ has at least one incoming arc or one outgoing
arc in $\W$.

\section{Generalised soundness is undecidable}
\label{sec:undecidability}
In this section, we prove the following result.

\begin{theorem}\label{thm:undecidability}
  The generalised soundness problem for reset workflow nets is
  undecidable.
\end{theorem}

We will give a reduction from the reachability problem in Minksy
machines. A \emph{(two-counter) Minsky machine} is a finite automaton
with two counters that can be incremented, decremented and
zero-tested. More formally, it is a pair $(Q, \Delta)$ where $Q$ is a
finite set of elements called \emph{control states}, and where $\Delta
\subseteq Q \times \{\inc{i}, \dec{i}, \zrt{i} : i \in \{1, 2\}\}
\times Q$ is a set of elements called \emph{transitions}. A
\emph{configuration} of such a machine is a triple $(p, a, b) \in Q
\times \N \times \N$ which we will write more concisely as
$p(\vec{v})$ where $\vec{v} = (a, b)$. For the following definition,
let $\neg i$ denote $3 - i$, i.e.\ the index of the other counter. A
transition $(p, \mathtt{oper}, q)$ allows to update a current
configuration in control state $p$ into a configuration in control
state $q$ with the expected semantic of $\mathtt{oper}$:
\begin{alignat*}{3}
  q(\vec{v}) &\trans{\mathmakebox[40pt][l]{(q, \inc{i}, q')}} q'(\vec{v}')\
  && \text{ if }
  \vec{v}'(i) = \vec{v}(i) + 1 &\text{ and } (*), \\
  q(\vec{v}) &\trans{\mathmakebox[40pt][l]{(q, \dec{i}, q')}} q'(\vec{v}')\
  && \text{ if }
  \vec{v}'(i) = \vec{v}(i) - 1 \geq 0 &\text{ and } (*), \\
  q(\vec{v}) &\trans{\mathmakebox[40pt][l]{(q, \zrt{i}, q')}} q'(\vec{v}')\
  && \text{ if }
  \vec{v}'(i) = \vec{v}(i) = 0 &\text{ and } (*),
\end{alignat*}
where ``$(*)$'' stands for ``$\vec{v}'(\neg i) = \vec{v}(\neg i)$''.

We write $q(\vec{v}) \trans{} q'(\vec{v}')$ if $q(\vec{v}) \trans{t}
q'(\vec{v}')$ for some $t \in \Delta$. We write $\trans{*}$ to denote
the reflexive-transitive closure of $\trans{}$. The
\emph{(control-state) reachability problem} asks, given $p, q \in Q$,
whether $p(\vec{0}) \trans{*} q(\vec{0})$. It is well known that this
problem is undecidable.

Let $q(\vec{v}) \trans{}_k q'(\vec{v}')$ denote the fact that
$q(\vec{v}) \trans{} q'(\vec{v}')$ and $0 \leq \vec{v}(i), \vec{v}'(i)
\leq k$ for both $i \in \{1, 2\}$. Let $\trans{*}_k$ denote the
reflexive-transitive closure of $\trans{}_k$. In words, $\trans{*}_k$
is the reachability relation where counters remain within
$[0..k]$. Clearly, $p(\vec{u}) \trans{*} q(\vec{v})$ holds iff there
exists $k \in \N$ such that $p(\vec{u}) \trans{*}_k q(\vec{v})$ holds.

\begin{proposition}\label{prop:minsky:resetpetri}
  Given a Minsky machine $\M = (Q, \Delta)$ and control states $\qsrc,
  \qtgt \in Q$, one can construct, in polynomial time, a reset Petri
  net $\PN$ with places $P \defeq Q \cup \{x_1, x_2, \budget{x_1},
  \budget{x_2}$\} such that the following holds for every $p, q \in Q$
  and $k \in \N$:
  \begin{itemize}
  \item $p(\vec{0}) \trans{*}_k q(\vec{0})$ holds in $\M$ iff $\{p
    \colon 1, \budget{x_1} \colon k, \budget{x_2} \colon k\} \trans{*}
    \{q \colon 1, \allowbreak \budget{x_1} \colon k, \allowbreak
    \budget{x_2} \colon k\}$ holds in $\PN$;

  \item $\{p \colon 1, \budget{x_1} \colon k, \budget{x_2} \colon k\}
    \trans{*} \m$ in $\PN$ implies $\m(Q) = 1$ and $\m(x_i) +
    \m(\budget{x_i}) \leq k$ for all $i \in \{1, 2\}$;

  \item Each node of $\PN$ is on some path from $\qsrc$ to $\qtgt$.
  \end{itemize}
\end{proposition}

\begin{proof}
  We use classical notions: budget places and weak simulation of
  zero-tests. More precisely, each counter $x_i$ of $\M$ is
  represented by two places in $\PN$: $x_i$ and
  $\budget{x_i}$. Initially, $x_i$ is empty and $\budget{x_i}$
  contains $k$ tokens. Whenever $x_i$ is incremented, $\budget{x_i}$
  is decremented, and vice versa. This forces $x_i$ to remain within
  $[0..k]$. Each zero-test of $\M$ is simulated by a reset of
  $x_i$. If a reset occurs whenver $x_i$ is empty, then nothing
  happens. However, if $\PN$ ``cheats'' and resets $x_i$ whenever it
  is non empty, then $\{x_i, \budget{x_i}\}$ now contains less than
  $k$ tokens and it will never be possible to increase that number back.

  More formally, let us define $\PN = (P, T, F, R)$. We set $P \defeq
  Q \cup \{x_1, x_2, \budget{x_1}, \budget{x_2}$\} and $T \defeq
  \Delta$. For each transition $t = (p, \mathtt{oper}, q) \in \Delta$,
  we add the following arcs to $\PN$.

  \begin{itemize}
  \item \emph{Case $\mathtt{oper} = \inc{i}$}. We move the token from
    $p$ to $q$, increment $x_i$ and decrement its dual:

    
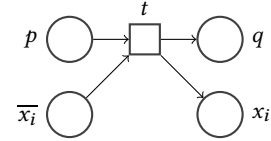
\begin{center}
      \begin{tikzpicture}[auto, node distance=1cm]
        \node[place, label=left:$p$] (p) {};
        \node[transition, right of=p, label=above:$t$] (t) {};
        \node[place, right of=t, label=right:$q$] (q) {};
        \node[place, below of=p, label=left:$\budget{x_i}$] (xib) {};
        \node[place, below of=q, label=right:$x_i$] (xi) {};

        \path[->]
        (p)   edge node {} (t)
        (t)   edge node {} (q)
        (t)   edge node {} (xi)
        (xib) edge node {} (t)
        ;
      \end{tikzpicture}
    \end{center}
    \medskip

 \item \emph{Case $\mathtt{oper} = \dec{i}$}. We move the token from
   $p$ to $q$, decrement $x_i$ and increment its dual:

    
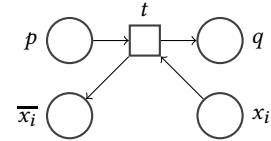
\begin{center}
      \begin{tikzpicture}[auto, node distance=1cm]
        \node[place, label=left:$p$] (p) {};
        \node[transition, right of=p, label=above:$t$] (t) {};
        \node[place, right of=t, label=right:$q$] (q) {};
        \node[place, below of=p, label=left:$\budget{x_i}$] (xib) {};
        \node[place, below of=q, label=right:$x_i$] (xi) {};

        \path[->]
        (p)  edge node {} (t)
        (t)  edge node {} (q)
        (t)  edge node {} (xib)
        (xi) edge node {} (t)
        ;
      \end{tikzpicture}
    \end{center}
    \medskip
    
  \item \emph{Case $\mathtt{oper} = \zrt{i}$}. We move the token from
    $p$ to $q$, reset $x_i$ and leave $\budget{x_i}$ unchanged:
    
    \smallskip

    
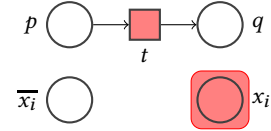
\begin{center}
      \begin{tikzpicture}[auto, node distance=1cm]
        \tikzstyle{ctrans} = [transition]

        \node[place, label=left:$p$] (p) {};
        
        \node[
          ctrans, right of=p, label=below:$t$,
          fill=colA!50
        ] (t) {};

        \node[place, right of=t, label=right:$q$] (q) {};
        \node[place, below of=p, label=left:$\budget{x_i}$] (xib) {};
        \node[place, below of=q, label=right:$x_i$] (xi) {};

        \path[->]
        (p)  edge node {} (t)
        (t)  edge node {} (q)
        ;

        \begin{pgfonlayer}{background}
          \node[fit=(xi),
            draw=colA, rounded corners, fill=colA!50, inner sep=2pt] {};
        \end{pgfonlayer}
      \end{tikzpicture}
    \end{center}
  \end{itemize}
  \medskip

  It is readily seen that starting from marking $\{\qsrc \colon 1,
  \allowbreak \budget{x_1} \colon k, \allowbreak \budget{x_2} \colon
  k\}$ there is always exactly one token in $Q$. Moreover, the two
  first types of transitions leave the number of tokens in $\{x_i,
  \budget{x_i}\}$ unchanged, while the third type of transitions may
  decreases the number of tokens in $\{x_i, \budget{x_i}\}$. So, the
  second item of the proposition holds.

  Let us explain why the first item holds.

  $\Rightarrow$) Assume $p(\vec{0}) \trans{\pi}_k q(\vec{0})$ holds in
  $\M$. We claim that
  \[
  \{p \colon 1, \budget{x_1} \colon k, \budget{x_2} \colon k\}
  \trans{\pi}
  \{q \colon 1, \budget{x_1} \colon k, \budget{x_2} \colon k\}
  \text{ holds in } \PN.
  \]
  Indeed, (i)~resets only occur on empty places, which maintains the
  invariant that $\{x_i, \budget{x_i}\}$ contains exactly $k$ tokens;
  (ii)~no increment or decrement is ever blocked in $\PN$ since we
  know that counters never exceed $k$ in $\M$.

  $\Leftarrow$) Assume that
  \[
  \{p \colon 1, \budget{x_1} \colon k, \budget{x_2} \colon k\}
  \trans{\pi}
  \{q \colon 1, \budget{x_1} \colon k, \budget{x_2} \colon k\}
  \text{ holds in } \PN.
  \]
  As each $\budget{x_i}$ starts and ends with $k$ tokens, this
  means that each reset that occured in $\pi$ did not consume any
  token. Thus, zero-tests were simulated faithfully. Consequently,
  $p(\vec{0}) \trans{\pi}_k q(\vec{0})$ holds in $\M$.

  It remains to prove the last item of the proposition, namely that
  each node of $\PN$ is on some path from $\qsrc$ to $\qtgt$. We
  preprocess $\M$ as follows:
  \begin{enumerate}
  \item each control state unreachable from $\qsrc$ in the underlying
    graph is removed from $\M$;\label{itm:rem:from}

  \item each control state that cannot reach $\qtgt$ in the underlying
    graph is removed from $\M$;\label{itm:rem:to}

  \item We add two new control states $\qsrc'$ and $r$ to $\M$, and
    these transitions:
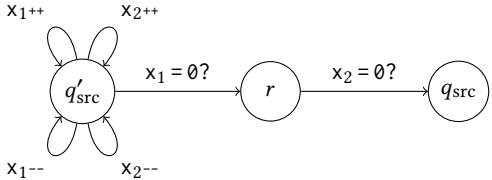
\label{itm:add}
    \medskip

    \begin{center}
      \begin{tikzpicture}[auto, node distance=2.5cm]
        \node[state] (qsrcp) {$\qsrc'$};
        \node[state, right of=qsrcp] (r) {$r$};
        \node[state, right of=r] (qsrc) {$\qsrc$};

        \path[->]
        (qsrcp) edge[loop, out=100, in=130, looseness=8] node[swap]
             {$\inc{1}$} (qsrcp)
        (qsrcp) edge[loop, out=-100, in=-130, looseness=8] node[]
             {$\dec{1}$} (qsrcp)

        (qsrcp) edge[loop, out=80, in=50, looseness=8] node[]
             {$\inc{2}$} (qsrcp)
        (qsrcp) edge[loop, out=-80, in=-50, looseness=8] node[swap]
             {$\dec{2}$} (qsrcp)

        (qsrcp) edge[] node {$\zrt{1}$} (r)
        (r) edge[] node {$\zrt{2}$} (qsrc)
        ;
      \end{tikzpicture}
    \end{center}
  \end{enumerate}
  The resulting machine $\M'$ is clearly equivalent,
  i.e.\ $\qsrc(\vec{0}) \trans{*}_k \qtgt(\vec{0})$ holds in $\M$ iff
  $\qsrc'(\vec{0}) \trans{*}_k \qtgt(\vec{0})$ holds in $\M'$.

  By the definition of the Petri net $\PN'$ obtained from $\M'$, each
  place of $Q$ is on some path from $\qsrc'$ to $\qtgt$ due to
  \eqref{itm:rem:from} and~\eqref{itm:rem:to}, and each place of
  $\{x_1, x_2, \budget{x_1}, \budget{x_2}\}$ as well since they can
  all reach $\qsrc$ in this fragment of $\PN'$ due to \eqref{itm:add}:
  \begin{center}
    \begin{tikzpicture}[auto, node distance=1.5cm, transform shape, scale=0.8]
      \tikzstyle{ctrans} = [transition]
      
      \node[place, label=above:$\qsrc'$] (qsrc) {};

      \node[
        ctrans, right of=qsrc,
        fill=colA!50
      ] (t1) {};

      \node[place, label=above:$r$, right of=t1] (r) {};

      \node[
        ctrans, right of=r,
        pattern=north east lines, pattern color=gray!75,
        preaction={fill=colB!50},
      ] (t2) {};

      \node[place, right of=t2, label=above:$\qsrc$] (qsrcp) {};

      \node[place, below of=qsrc, label=above:$x_2$] (x2) {};
      \node[place, below of=x2, label=below:$\budget{x_2}$] (x2b) {};

      \node[transition, right of=x2]  (s1) {};
      \node[transition, right of=x2b] (s2) {};

      \node[transition, left of=x2]  (u1) {};
      \node[transition, left of=x2b] (u2) {};

      \node[place, left of=u1, label=above:$x_1$] (x1) {};
      \node[place, left of=u2, label=below:$\budget{x_1}$] (x1b) {};
        
      \path[->]
      (p)  edge[bend left=5] node {} (s1)
      (s1) edge[bend left=5] node {} (p)
      (s1) edge node {} (x2b)
      (x2) edge node {} (s1)

      (p)   edge[bend left=5] node {} (s2)
      (s2)  edge[bend left=5] node {} (p)
      (s2)  edge node {} (x2)
      (x2b) edge node {} (s2)

      (p)  edge[bend left=5] node {} (u1)
      (u1) edge[bend left=5] node {} (p)
      (u1) edge node {} (x1b)
      (x1) edge node {} (u1)

      (p)   edge[bend left=5] node {} (u2)
      (u2)  edge[bend left=5] node {} (p)
      (u2)  edge node {} (x1)
      (x1b) edge node {} (u2)

      (qsrc) edge node {} (t1)
      (t1)   edge node {} (r)

      (r)  edge node {} (t2)
      (t2) edge node {} (qsrcp)
      ;

      \begin{pgfonlayer}{background}
        \node[fit=(x1),
          draw=colA, rounded corners, fill=colA!50, inner sep=2pt] {};
        
        \node[fit=(x2), pattern=north east lines, pattern color=gray!75,
          draw=colB, rounded corners, preaction={fill=colB!50},
          inner sep=2pt] {};
      \end{pgfonlayer}      
    \end{tikzpicture}
    \vspace{-20pt}
  \end{center}
\end{proof}

The following proposition establishes \Cref{thm:undecidability}.

\begin{proposition}
  Given a Minsky machine $\M = (Q, \Delta)$ and control states $\qsrc,
  \qtgt \in Q$, one can construct, in polynomial time, a reset
  workflow net $\W$ such that $\W$ is generalised sound iff
  $\qsrc(\vec{0}) \not\trans{*} \qtgt(\vec{0})$ holds in $\M$.
\end{proposition}

\begin{proof}
  We first sketch the reset workflow net $\W$, describe the
  construction, and show that it is indeed a reset workflow net.

  Let $\PN$ be the reset Petri net given by \Cref{prop:minsky:resetpetri}
  from $\M$. The reset workflow net $\W$ consists of $\PN$ together
  with the extra places $\{i, r, f\}$ and transitions $\{t_1, t_2,
  t_3\}$. \Cref{fig:worflow:reduction} depicts $\W$ where the solid
  red part corresponds to $\PN$. Places $i$ and $f$ are respectively
  the initial and final places of $\W$. When either of $t_1$, $t_2$ or
  $t_3$ is fired, all places from the corresponding colored area is
  reset. By \Cref{prop:minsky:resetpetri}, each node of $\W$ is on
  some path from $i$ to $f$.
    
  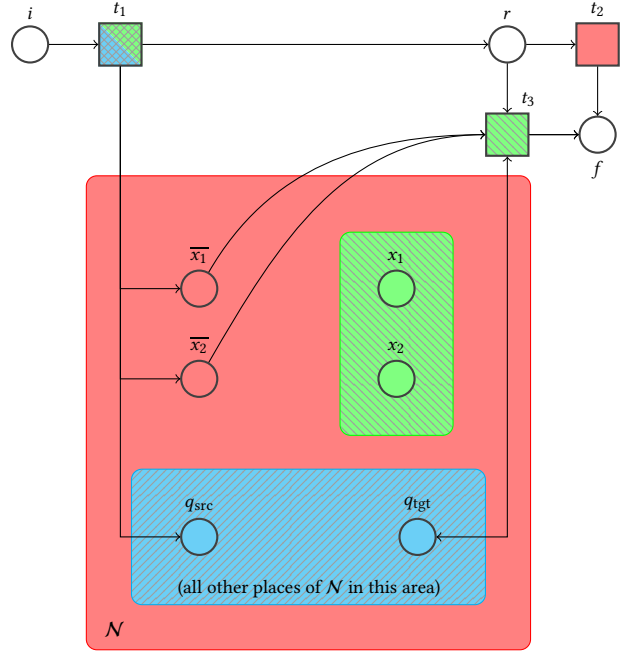
\begin{figure}[h]
    \centering
    \begin{tikzpicture}[
  auto, node distance=1.5cm, transform shape, scale=0.8]
  
  \tikzstyle{ctrans} = [transition, minimum size=20pt]

  \node[place, label=above:$i$]   (i)  {};

  \node[ctrans, label=above:$t_1$, right of=i,
    pattern=crosshatch, pattern color=gray!75] (t1) {};

  \begin{pgfonlayer}{background}
    \fill[colC!50]
    (t1.north west) -- (t1.north east) -- (t1.south east) -- cycle;

    \fill[colB!50]
    (t1.north west) -- (t1.south west) -- (t1.south east) -- cycle;
  \end{pgfonlayer}
  
  \node[place, label=above:$\budget{x_1}$, below right of=t1, xshift=0.25cm, yshift=-3cm] (x1) {};
  \node[place, label=above:$\budget{x_2}$, below of=x1] (x2) {};

  \node[place, label=above:$x_1$, right=2.65cm of x1,
    fill=colC!50] (x1o) {};
  \node[place, label=above:$x_2$, right=2.65cm of x2,
    fill=colC!50] (x2o) {};

  \node[place, label=above:$\qsrc$, below=2cm of x2, fill=colB!50] (qsrc) {};
  \node[place, label=above:$\qtgt$, right=3cm of qsrc, fill=colB!50] (qtgt) {};

  \node[place, label=above:$r$, right=5.75cm of t1] (r) {};

  \node[ctrans, label={[shift={(10pt, 0)}]$t_3$}, below of=r,
    pattern=north west lines, pattern color=gray!75,
    preaction={fill=colC!50}] (t3) {};

  \node[ctrans, label=above:$t_2$, right of=r, fill=colA!50] (t2) {};

  \node[place, label=below:$f$, right of=t3] (f) {};

  \node[below=8pt of qsrc, xshift=52pt] {(all other places of $\PN$ in this area)};

  \node[below=28pt of qsrc, xshift=-40pt] {\large$\PN$};

  \draw[->] (t1) |- (x1);
  \draw[->] (t1) |- (x2);
  \draw[->] (t1) |- (qsrc);

  \path[->]
  (i)  edge node {} (t1)
  (t1) edge node {} (r)

  (r)  edge node {} (t2)
  (t2) edge node {} (f)

  (r)  edge node {} (t3)
  (t3) edge node {} (f)

  (t3) edge node {} (f)

  (x1) edge[out=60, in=180] node {} (t3)
  (x2) edge[out=60, in=180] node {} (t3)
  ;

  \draw[<->] (t3) |- (qtgt);

  \begin{pgfonlayer}{background}
    \node[fit=(x1)(x1o)(x2)(x2o)(qsrc)(qtgt), 
      draw=colA, rounded corners, fill=colA!50, inner sep=1.25cm] {};

    \node[fit=(x1o)(x2o), pattern=north west lines, pattern color=gray!75,
      draw=colC, rounded corners, preaction={fill=colC!50},
      inner sep=0.5cm] {};

    \node[fit=(qsrc)(qtgt), pattern=north east lines, pattern color=gray!75,
      draw=colB, rounded corners, preaction={fill=colB!50},
      inner sep=0.65cm] {};
  \end{pgfonlayer}
\end{tikzpicture}    
    \caption{The reset workflow net $\W$. The bidirectional arc between
      $\qtgt$ and $t_3$ represents two arcs: one in each
      direction. Each transition $t_i$ resets the places within the
      part of the corresponding pattern and color. More precisely,
      $t_1$ resets all places of $\PN$ except for $\{\budget{x_1},
      \budget{x_2}\}$; $t_2$ resets all places of $\PN$; and $t_3$
      resets $\{x_1, x_2\}$.}
    \label{fig:worflow:reduction}
  \end{figure}

  
  Let us now show that there exists $k \in \N$ such that
  $\qsrc(\vec{0}) \trans{*}_k \qtgt(\vec{0})$ in $\M$ iff $\W$ is not
  generalised sound.

  \medskip\noindent $\Rightarrow$) Let $k \in \N$ be such that
  $\qsrc(\vec{0}) \trans{*}_k \qtgt(\vec{0})$ in $\M$. By
  \Cref{prop:minsky:resetpetri}, we have
  \[
  \{\qsrc \colon 1, \budget{x_1} \colon k, \budget{x_2} \colon k\}
  \trans{\pi}
  \{\qtgt \colon 1, \budget{x_1} \colon k, \budget{x_2} \colon k\}
  \text{ in } \PN.
  \]
  Thus, the following holds in $\W$:
  \begin{align*}
  \{i \colon k\}
  &\trans{t_1^k}
  \{\qsrc \colon 1, \budget{x_1} \colon k, \budget{x_2} \colon k, r
  \colon k\} \\
  &\trans{\pi}
  \{\qtgt \colon 1, \budget{x_1} \colon k, \budget{x_2} \colon k, r
  \colon k\} \\
  &\trans{t_3^k}
  \{\qtgt \colon 1, f \colon k\}.
  \end{align*}
  The latter marking witnesses $k$-unsoundness of $\W$. Indeed, by
  \Cref{prop:minsky:resetpetri}, the subset $Q$ of places of $\PN$,
  that corresponds to the control states of $\M$, cannot be emptied.

  \medskip\noindent $\Leftarrow$) Let $\{i \colon k\} \trans{\pi} \m$
  witness $k$-unsoundness of $\W$, where $k \in \Npos$. First, we show
  that $\m(f) = k$. Note that in any marking $\m'$ reachable from $\{i
  \colon k\}$, we have $\m'(i) + \m'(r) + \m'(f) = k$. Thus, $\m(f)
  \leq k$. For the sake of contradiction, suppose that $\m(f) <
  k$. Let $k_i \defeq \m(i)$ and $k_r \defeq \m(r)$. By the previous
  equality, we have $k_i + k_r > 0$. The following holds in $\W$:
  \[
  \m
  \trans{t_1^{k_i}\ t_2^{k_i + k_r}}
  \{f \colon k\}.
  \]
  Indeed, as $k_i + k_r > 0$, transition $t_2$ is fired at least once
  and hence $\PN$ is emptied (the solid red area of
  \Cref{fig:worflow:reduction}). This contradicts the assumption that
  $\m$ witnesses $k$-unsoundness.

  We have established that $\m(f) = k$, which implies that $\m(i) = 0$
  and $\m(r) = 0$. Without loss of generality we can assume that $k$
  is minimal i.e.\ $\W$ is $\ell$-sound for every $\ell < k$.

  \begin{claim}\label{claim:t2}
    Transition $t_2$ does not appear in $\pi$.
  \end{claim}

  For the sake of contradiction, suppose the claim does not hold. We
  split run $\pi$ into $\pi = \pi_1 \pi_2$ where $\pi_2$ is a maximal
  suffix that does not contain transition $t_2$. We have $\{i \colon
  k\} \trans{\pi_1} \m_1 \trans{\pi_2} \m$ for some marking $\m_1$.

  Let us compare the number of occurrences of $t_1$ and $t_3$ in
  $\pi_2$. As $t_2$ is the last transition of $\pi_1$, we have
  $\m_1(p) = 0$ for every place $p$ of $\PN$. By
  \Cref{prop:minsky:resetpetri}, the number of tokens in $\{x_i,
  \budget{x_i}\}$ cannot increase by firing transitions of
  $\PN$. Thus, we must have $|\pi_2|_{t_1} \geq |\pi_2|_{t_3}$. Since
  $\m(r) = 0$, we must have $|\pi_2|_{t_1} \leq
  |\pi_2|_{t_3}$. Consequently, $|\pi_2|_{t_1} =
  |\pi_2|_{t_3}$. Moreover, $|\pi_2|_{t_1} > 0$ as otherwise $\m = \{f
  \colon k\}$, which is obviously not a witness of $k$-unsoundness.

  From this and $\m(r) = 0$, we conclude that $\m_1(r) = 0$. This
  means that all places, except possibly $i$ and $f$, are empty in
  $\m_1$. Since $\m(f) = k$ and $|\pi_2|_{t_1} = |\pi_2|_{t_3} > 0$,
  the marking $\m_1$ is of the form $\m_1 = \{i \colon k - \ell, f
  \colon \ell\}$ where $0 < \ell < k$.

  As no transition consumes from $f$, we have $\{i \colon k - \ell\}
  \trans{\pi_2} \m'$ where $\m' \defeq \m - \{f \colon \ell\}$. Since
  $\m'(f) = \m(f) - \ell = k - \ell$, we conclude that $\W$ is not $(k
  - \ell)$-sound as $\m'$ is also a witness of unsoundness. This
  contradicts the minimality of $k$. So, \Cref{claim:t2} holds as
  desired.

  \begin{claim}\label{clm:3tokens}
    In run $\pi$,
    \begin{enumerate}
    \item Every occurrence of $t_1$ does not remove any token from
      places $\{x_1, x_2\}$;
      
    \item Every occurrence of $t_3$ consumes (exactly) two tokens from
      places $\{x_1, x_2, \budget{x_1}, \budget{x_2}\}$.
    \end{enumerate}
  \end{claim}
  
  From \Cref{claim:t2}, we have $|\pi|_{t_1} =
  |\pi|_{t_3}$. Moreover, each occurrence of transition $t_3$ removes
  at least two tokens from $\{x_1, x_2, \budget{x_1}, \budget{x_2}\}$,
  and each occurrence of transition $t_1$ adds at most two tokens to
  $\{x_1, x_2, \budget{x_1}, \budget{x_2}\}$. Hence,
  \Cref{clm:3tokens} holds as desired.

  \medskip
  
  Now, let us split $\pi$ into $\pi = \pi_3 \pi_4 \pi_5$ where
  $\pi_3\pi_4$ is the longest prefix of $\pi$ without $t_3$, and
  $\pi_4$ is a maximal suffix of $\pi_3 \pi_4$ that does not contain
  transition $t_1$. Let $\m_3$ and $\m_4$ be the markings such that
  $\{i \colon k\} \trans{\pi_3} \m_3 \trans{\pi_4} \m_4$. We have
  $\m_4(x_1) = \m_4(x_2) = 0$, as otherwise the first occurrence of
  $t_3$ removes at least three tokens from places $\{x_1, x_2,
  \budget{x_1}, \budget{x_2}\}$, which contradicts
  \Cref{clm:3tokens}. A similar argument shows that $\m_3(x_1) =
  \m_3(x_2) = 0$.

  Since the last transition of $\pi_3$ is $t_1$, since $\pi_3$
  contains no occurrence of $\{t_2, t_3\}$, since $\pi_4$ contains no
  occurrence of $\{t_1, t_2, t_3\}$, and since the first transition of
  $\pi_5$ is $t_3$, there exists $a > 0$ such that
  \begin{itemize}
  \item $\m_3 = \{i \colon k - a, r \colon a, \qsrc \colon 1,
    \budget{x_1} \colon a, \budget{x_2} \colon a\}$, and
    
  \item $\m_4 = \{i \colon k - a, r \colon a, \qtgt \colon 1,
    \budget{x_1} \colon a, \budget{x_2} \colon a\}$.
  \end{itemize} 
  This implies that $\m_3 \trans{\pi_4} \m_4$ induces a run
  $\qsrc(\vec{0}) \trans{*}_a \qtgt(\vec{0})$ of Minsky machine $\M$.
\end{proof}

\section{In between soundness}
\label{sec:between}
We start this section by discussing properties that reset workflow
nets can satisfy. So far, we mostly discussed generalised soundness
and $1$-soundness, or more generally, $k$-soundness for any $k \in
\Npos$.
We say that a reset workflow net is \emph{up-to-$k$-sound}~\cite[Definition~24]{phdthesisToorn} if it is $j$-sound for all $j \in [1..k]$. Observe that up-to-$1$-soundness is equivalent to $1$-soundness.
By definition, a reset workflow net which is generalised
soundness is also up-to-$k$ sound (for any $k \in
\Npos$).

\begin{figure}
  \begin{tikzpicture}[scale=0.85, transform shape]
  \draw[fill = gray!90] (0,0) rectangle (8,8);
  
  \node at (1.85,7.5) {\bf All reset workflow nets};
  
  \draw[fill = gray!60] (0.5,0) parabola bend (4,7) (7.5,0);
  \draw[pattern = horizontal lines light gray]  (1,0) parabola bend (4,5) (7,0);
  \draw[fill = white] (2.3,0) parabola bend (4,2.5) (5.7,0);
  
  \draw[] (2.3,0) parabola bend (4,2.5) (5.7,0);
  
  \draw[] (1,0) parabola bend (4,5) (7,0);
  
  \draw[] (0.5,0) parabola bend (4,7) (7.5,0);
  
  \node[align=center] at (4,0.5) {\bf generalised sound \\ \bf reset workflow nets};
  \node[align=center] at (4,3.3) {\bf $k$-in-between sound \\ \bf reset workflow nets};
  \node[align=center] at (4,5.75) {\bf up-to-$k$-sound \\ \bf reset workflow nets};
\end{tikzpicture}
  \caption{Classes of reset workflow nets: generalised sound,
    up-to-$k$-sound, $k$-in-between sound and all reset workflow
    nets. Properties with lighter colors also satisfy darker colored
    properties. For example, the class of generalised sound reset
    workflow nets is the most restrictive and it is contained in all
    other classes. Note that ``$k$-in-between sound'' is a class of
    properties, not a single one. So, in the figure, one should
    think of each horizontal line as one of
    these properties.}\label{fig:classes}
\end{figure}
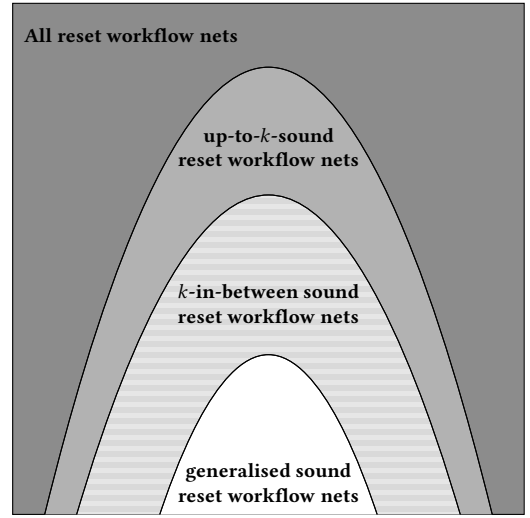

We say that a property $\mathcal{P}$, of reset workflow nets, is \emph{$k$-in-between sound} if: all generalised sound workflow nets satisfy $\mathcal{P}$; and every workflow net that satisfies $\mathcal{P}$ is up-to-$k$-sound. \Cref{fig:classes} depicts all of the aforementioned classes.


Our first main result is as follows.

\begin{theorem}\label{theorem:between}
For every $k \in \Npos$, there exists a decidable $k$-in-between sound property $\mathcal{P}_k$ of reset workflow nets. More precisely: given as input $k \in \Npos$ and a reset workflow net, there is an algorithm deciding $\mathcal{P}_k$ running in Ackermannian time.
\end{theorem}


\begin{remark}
In particular, given a reset workflow net $\W$ and $k \in \Npos$, there is an algorithm that correctly outputs: either that $\W$ is up-to-$k$-sound; or that $\W$ is not
generalised sound.
More precisely, if $\W$ is not up-to-$k$-sound, then it outputs that $\W$ is not generalised sound; if $\W$ is generalised sound, then it outputs that it is up-to-$k$-sound; otherwise, it can output either of the two properties (both hold).
\end{remark}

We deliberately postpone the formal definition of property
$\mathcal{P}_k$ as it is technical. Instead, we give intuition on
where it comes from. First observe that soundness is a conjunction of
three simpler informal properties: {\it (i)}~It is impossible to
strictly cover the final marking; {\it (ii)}~It is impossible to reach
markings with tokens only in the final place, but with insufficiently
many of them; and {\it (iii)}~It is impossible to reach a marking that
has tokens in places other than $\final$ and from which it is
impossible to produce more tokens in $\final$.

Now, on the one hand, to decide $1$-soundness for~{\it (i)} we can check if it is
impossible to strictly cover $\{\final \colon 1\}$ as coverability is
decidable. Property~{\it (ii)} holds as the workflow net cannot be
emptied, indeed, firing any transition always produces tokens. So, the
reason for undecidability of $1$-soundness is the hardness of
property~{\it (iii)}. On the other hand, for generalised soundness, it
can be proven that property~{\it (ii)} is decidable, and, assuming
that~{\it (i)} holds {\it (iii)}~is decidable, so the essential reason
of the undecidability of the generalised soundness is~{\it (i)}. The
idea behind $\mathcal{P}_k$ is to exploit the decidable parts, which intuitively cover all three parts of soundness.
A bit more precisely, we combine the test for~{\it (ii)} for
generalised soundness, the test for~{\it (i)} for $k$-soundness,
and~{\it (iii)} for generalised soundness assuming that~{\it (i)}
holds. This mixture of decidable properties gives rise to a relation
that is $k$-in-between.

Surprisingly, the property $\mathcal{P}_k$ defined in this way for
sufficiently large $k$ coincides with to up-to-$k$ soundness. More
formally:


\begin{theorem}\label{theorem:upto}
  Given a nonredundant reset workflow net $\W$, there is a computable number $k'
  \in \Npos$ (Ackermannian in the size of $\W$) such that, for all $k
  \ge k'$, $\W$ satisfies $\mathcal{P}_k$ iff $\W$ is
  up-to-$k$ sound.
\end{theorem}

The rest of the section is organized as follows. First, in
\Cref{ssec:redundancy}, we introduce the notion of ``nonredundancy'',
which allows us to identify unimportant places and transitions. Next,
in \Cref{ssec:skeleton}, we define the ``skeleton'' of a reset
workflow net, which is a crucial object in deciding property~${\it
  (ii)}$ for generalised soundness. In \Cref{sec:skeleton1}, we show
that the skeleton of a nonredundant reset workflow net is a workflow
net (without resets). In \Cref{sec:skeleton2}, we show that if a
nonredundant reset workflow net is generalised sound, then its
skeleton is also generalised sound. This last property implies
property~${\it (ii)}$ for generalised soundness. Finally,
\Cref{ssec:thm:proofs} makes use of these results to prove
\Cref{theorem:between,theorem:upto}; Here, in particular, we show how
property~${\it (iii)}$ can be tested assuming that~${\it (i)}$ holds,
although it is not given explicitly to keep the argument shorter.

\subsection{Nonredundancy}\label{ssec:redundancy}

We provide a technical definition of redundancy, following similar
definitions for workflow nets (without
resets)~\cite{HeeSV04}. Intuitively, it allows to ignore useless
places and transitions from the net without changing its set of
reachable markings.

Formally, given a reset workflow net $\W$, we say that
\begin{enumerate}
  \item a place $p$ is \emph{nonredundant} if there exist $k \in \N$
    and a marking $\m$ such that $\{\initial \colon k\} \trans{*} \m$
    and $\m(p) \geq 1$;\label{itm:redund:a}
    
  \item a transition $t$ is \emph{nonredundant} if there exist $k \in
    \N$ and a marking $\m$ such that $\{\initial \colon k\} \trans{*}
    \m$ and $t$ is enabled in $\m$.\label{itm:redund:b}
\end{enumerate}
If \Cref{itm:redund:a} does not hold for a place $p$, then we say that
$p$ is \emph{redundant}, and likewise for a transition
$t$.

\begin{proposition}\label{proposition:redundancy}
  Given a reset workflow net $\W$, one can compute its set of
  redundant places and transitions. The procedure works in
  Ackermannian time. Further, for every nonredundant transition $t$
  and place $p$, there are numbers $k_t, k_p \in \N$ bounded
  Ackermannianly and runs of length at most Ackermannian such that
  $\{\initial \colon k_t\} \trans{*} \upclosure{(\pre{t})}$ and
  $\{\initial \colon k_p\} \trans{*} \upclosure{\{p \colon 1\}}$.
\end{proposition}

\begin{proof}
  Let us show how to compute the set of nonredundant places and
  transitions.  Recall that $C_{\m'}$ is the upward closed set of all
  markings $\m$ from which it is possible to cover the marking $\m'$;
  one may compute the set of minimal elements of $C_{\m'}$ with the
  backward coverability algorithm. As mentioned in the preliminaries,
  every upward closed subset of $\N^P$ is equal to a finite union of
  elements $\upclosure{\vec{x}_i}$ with $\vec{x}_i \in \N^P$.

  From the definition of nonredundancy, a place $p$ is nonredundant
  iff there exists some $k_p$ such that $\{\initial \colon k_p\}$ is
  in $C_{\{p \colon 1\}}$; this is decidable, using the backwards
  coverability algorithm~\cite{LazicS21}, and one may compute such a
  $k_p$. Similarly, a transition $t$ is nonredundant iff there exists
  $k_t$ such that $\{\initial \colon k_t\}$ is in $C_{\pre{t}}$, and
  this is also decidable and $k_t$ is still computable. Now, let $K$
  be the sum of all $k_p$ and $k_t$ for $p\in P$ and $t\in T$. Observe
  that, for all markings $\m \in \bigcup_{p\in P, t\in T} \{\{p \colon
  1\} + \pre{t}\}$, there is a run, from the initial marking
  $\{\initial \colon K\}$, that covers $\m$. The Ackermannian bounds
  follow from the bounds on the coverability
  problem~\cite{FigueiraFSS11}.
\end{proof}

We make the following observation on nonredundant transitions.

\begin{claim}\label{claim:noreset_if}
  Let $\W$ be a reset workflow net which is generalised
    sound. Nonredundant transitions cannot reset either $i$ or $f$.
  Moreover, this claim still holds if we relax the requirement of $\W$
  being generalised sound to a weaker one, namely to ``$\W$ is
  up-to-$K$ sound'', where number $K$ is at most Ackermannian.
\end{claim}

\begin{claimproof}
  For the final place $f$, the claim follows by definition. Let us
  consider the case of the initial place $i$. Towards a contradiction,
  suppose that $\W$ has a nonredundant transition $t$ such that
  $t$ resets place $i$. By \Cref{proposition:redundancy}, there exists
  a number $k \in \Npos$ which is at most Ackermannian such that
  \[
  \{\initial \colon k\} \trans{\pi } \m' \trans{t} \m,
  \]
  for some run $\pi$ that does not use transition $t$, and some
  markings $\m'$ and $\m$. As $k$ is Ackermannian, we can safely
  assume that $k < K$.  Thus, the following holds for $K > k$:
  \[
  \{\initial \colon K\} \trans{\pi} \m' + \{\initial \colon x\} \trans{t} \m,
  \]
  for $x \ge K-k > 0$.
  Note that $t$ resets place $\initial$, and that $k$ is defined in such a way that at least one token is lost. By generalised
  soundness, or up-to-$K$ soundness of $\W$, we have $\m \trans{*} \{f
  \colon k\}$ and $\m \trans{*} \{f \colon K\}$. This contradicts
  generalised soundness of $\W$ and up-to-$K$ soundness.
\end{claimproof}

\subsection{Skeletons of reset workflow nets}\label{ssec:skeleton}

In this subsection, we consider a relaxation of generalised
soundness. Given a reset workflow net $\W$, we will define a
  workflow net $\W^s$ obtained by removing redundancy and resetable
  places. Intuitively, generalised soundness of $\W$ should imply
generalised soundness of $\W^s$. As we shall see, this requires some
work. We start by introducing the notation.

Let $\W = (P, T, F, R, i, f)$ be a reset workflow net. We say that
place $p \in P$ is \emph{resetable} if there exists a
nonredundant transition $t \in T$ such that $(p, t) \in
R$. The \emph{skeleton} of a reset workflow net $\W$ is the Petri
  net obtained by removing redundant places, redundant transitions and
  resetable places, and next removing isolated transitions (as
defined in \Cref{ssec:subnets} of the preliminaries).

We denote this Petri net by $\W^s = (P^s, T^s, F^s)$. If there are well-defined initial and final places $i^s, f^s$ such that $(P^s, T^s, F^s, i^s, f^s)$ is a workflow net, then, slightly abusing the notation, we write that $\W^s = (P^s, T^s, F^s, i^s, f^s)$ is a workflow net.
We will devote the forthcoming \Cref{sec:skeleton1,sec:skeleton2} to proving the following.

\begin{proposition}\label{proposition:skeleton}
Let $\W$ be a reset workflow net which is generalised sound. It is the case that
\begin{enumerate}
 \item\label{skeleton1} the skeleton $\W^s$ is a workflow net;
 \item\label{skeleton2} $\W^s$ is also generalised sound.
\end{enumerate}
Moreover, this claim still holds if we relax the requirement of $\W$
being generalised sound to a weaker one, namely to ``$\W$ is up-to-$K$
sound'', where number $K$ is at most Ackermannian.
\end{proposition}


Note that generalised soundness for workflow nets (without resets) is decidable~\cite{HeeSV04}, and belongs to PSPACE~\cite{BlondinMO22}. So, \Cref{proposition:skeleton} implies that generalised soundness for the skeleton workflow net is a decidable relaxation of generalised soundness for $\W$.

Before we prove \Cref{proposition:skeleton}, we need a definition that allows us to associate markings in $\W$ with markings in $\W^s$.

The function $\Res{NoArg} \colon \N^{\Places} \to \N^{\Places}$ is defined by
\[
\Res{\vec{x}}(p)  =
\begin{cases}
  \vec{x}(p) & \text{if } p \in P^s, \\
           0 & \text{otherwise}.
\end{cases}
\]
Note that $\Res{\vec{x}}$ is equal to applying $\Rest{t}{\cdot}$ to $\vec{x}$ for each nonredundant transition $t$ (in any order).
\Cref{lemma:reset_marking2} will provide some intuition on why $\Res{\vec{x}}$ corresponds to a marking in $\W^s$.

As we often consider runs in the reset workflow $\W$, as well as runs in its skeleton $\W^s$, we introduce the notation ${\transS{}}$ to denote runs specifically in $\W^s$.

\begin{lemma}\label{lemma:reset_marking}
  Let $\W$ be a reset workflow net which is generalised sound. There
  exist $z \in \N$ and a run $\{i \colon z\} \trans{\zeta} \{f \colon
  z\}$, where $\zeta$ contains each nonredundant transition of
    $\W$. Moreover, $z$ and $|\zeta|$ are at most Ackermannian.

  Further, the claim still holds if we relax the assumption of $\W$
  being generalised sound to a weaker one, namely to ``$\W$ is
  up-to-$K$ sound'', where $K$ is at most Ackermannian (the same as
  the bound on $z$).
\end{lemma}

\begin{proof}
  For every nonredundant transition $t$ of $\W$, there is a run firing $t$, i.e.\ $\{i \colon z_t\} \trans{\rho_t t} \m_t$ for some $z_t \in \N$ and some marking $\m_t$. We define $z \defeq \sum_{t \in T} z_t$. Let $\rho$ be the concatenation (in any order) of runs $\rho_t t$. Since $\W$ is up-to-$K$ sound, by \Cref{claim:noreset_if}, the initial place $i$ cannot be reset by any nonredundant transition of $\W$. Thus, $\{i \colon z\} \trans{\rho} \m'$ for some marking $\m'$. Since $\W$ is generalised sound (or up-to-$z$-sound, as $z\le K$), there is a run $\delta$ such that $\m' \trans{\delta} \{f \colon z\}$. We conclude by taking $\zeta \defeq \rho \delta$.
 The Ackermannian bounds on $\rho$ follow from \Cref{proposition:redundancy}. The Ackermannian bounds on $\delta$ follow from a bound on the maximal length of the coverability run in the reset Petri net (\Cref{lem:boundOnCoverability} with target marking $\{f \colon z\}$).
 Thus, $|\zeta| = |\rho\delta|$ is at most Ackermannian.
\end{proof}

\begin{remark}\label{remark:reset_marking}
From the proof of \Cref{lemma:reset_marking}, we can see that if the initial place of $\W$ cannot be reset, then
there exist $z \in \N$ and $\{i \colon z\} \trans{\zeta} \m$, for some marking $\m$, that contains each nonredundant transition.
\end{remark}

We say that a marking $\m \in \N^P$ of $\W$ is \emph{reachable} if
  there exists $k \in \Npos$ such that $\{i \colon k\} \trans{*} \m$.

\begin{lemma}\label{lemma:reset_marking2}
  Let $\zeta$ be as in \Cref{lemma:reset_marking}, and let $\m$ be a
    reachable marking of $\W$. It is the case that $\{i \colon z\} +
  \m \trans{\zeta} \{f \colon z\} + \Res{\m}$.
\end{lemma}

\begin{proof}
  We have $\{i \colon z\} + \m \trans{\zeta} \n$ for some marking
  $\n$.  We need to prove that $\n = \{f \colon z\} + \Res{\m}$. It is
  easy to see that $\n(p) = (\{f \colon z\} + \Res{\m})(p)$ for all $p
  \in P^s$ as these places are not resetable, and the effect on
  them is the sum of effects on all transitions in $\zeta$. Thus, we
  need to prove that $\n(p) = 0$ for all $p \notin P^s$.

  Let $p \notin P^s$. If $p$ is redundant, then $\n(p) = \m(p) = 0$
  since $p$ is not marked in any reachable marking. Further, $f$
  cannot be reset by \Cref{claim:noreset_if}, and $p \neq f$ as $f$ is
  nonredundant by generalised soundness (or up-to-$K$ soundness) of
  $\W$.

  Thus, we may assume that $p$ is a nonredundant place, which is
    reset by a nonredundant transition $t$. By definition, we can
  decompose $\zeta$ as $\zeta = \rho_1 t \rho_2$.  Let $\{i \colon z\}
  \trans{\rho_1 t} \n_1$ and $\{i \colon z\} + \m \trans{\rho_1 t}
  \n_2$. We know that $\n_1(p) = \n_2(p) = 0$. Consider two runs, with the same sequence of transitions: $\rho_2$, but different starting points: $\n_1$ and $\n_2$. By induction, one can
  prove that the markings in both runs will always have the same value in $p$. Indeed, they are the same initially, and the same sequence of transitions are applied afterwards. Thus, in the end
  $\n(p) = (\{f \colon z\} + \Res{\m})(p) = 0$ as required.
\end{proof}

The next two subsections are devoted to proving the two items of
\Cref{proposition:skeleton}.

\subsection{Skeletons are workflow nets}\label{sec:skeleton1}

In this subsection, we prove the first item of \Cref{proposition:skeleton}.
We fix a reset workflow net $\W = (P, T, F, R, i, f)$ as in the statement of \Cref{proposition:skeleton} and its skeleton $(P^s, T^s, F^s)$. Overall, our aim is to prove that $\W^s = (P^s, T^s, F^s, i,f)$ is a well-defined workflow net.

It follows directly from \Cref{claim:noreset_if} that no nonredundant transition of $\W$ resets $i$ or $f$. Thus, the following holds.

\begin{lemma}
  It is the case that $i, f \in P^s$.
\end{lemma}

The following claim is useful and trivial by definition.

\begin{claim}\label{claim:projection}
Consider a run $\m \trans{\rho} \m'$ of $\W$. Let $\m_s$, $\m'_s \in \N^{P^s}$ be the markings obtained by projecting $\m$ and $\m'$ onto $P^s$; and let $\rho_s$ be the run in $\W^s$ obtained from $\rho$ (i.e. transitions are restricted to $P^s$, and possibly isolated transitions are removed). It is the case that $\m_s \transS{\rho_s} \m'_s$.
\end{claim}

We do not know yet whether $\W^s$ is a workflow net, but the definition of nonredundancy makes sense for $\W^s$ even if it is not a workflow net. Below, we note that nonredundancy of places in $\W$ easily implies nonredundancy of places in $\W^s$.

\begin{claim}\label{claim:nonredundant_skeleton}
For every place $p \in P^s$, there exist $k \in \N$, with $k \le K$, and a run $\{i \colon k\} \transS{*} \m_s$ such that $\m_s(p) > 0$ and $\m_s \transS{*} \{\final \colon k\}$.
\end{claim}

\begin{claimproof}
By definition of $\W^s$, place $p$ is nonredundant in $\W$, and hence we have $\{\initial \colon k\} \trans{*} \m$ where $\m(p) > 0$. Since $\W$ is generalised sound or up-to-$k$ sound for $k\le K$, there is a run $\m \trans{*} \{\final \colon k\}$.
These two runs and \Cref{claim:projection} conclude the proof.
\end{claimproof}

To show \Cref{proposition:skeleton}~\eqref{skeleton1}, it remains to prove that all places from $P^s$ and all transitions from $T^s$ are on a path from $i$ to $f$ in $\W^s$. The following lemma reduces the problem to checking this property for places only.

\begin{lemma}\label{lemma:transition_path}
Let $t \in T^s$. It is the case that $\pre{t} \neq \emptyset$ and $\post{t} \neq \emptyset$.
\end{lemma}

We will need the following technical claim. Note that although its
statement deals with places $P^s \subseteq P$ of the skeleton, the claim
deals with runs from $\W$.

\begin{claim}\label{claim:badconf}
Let $p \in P^s$. There is no $k \in \N,$ $k\le K$ such that $\W$ has two
  runs $\{i \colon k\} \trans{\pi} \m$ and $\{i \colon k\}
  \trans{\pi'} \m'$ with $\m' \geq \m + \{p \colon 1\}$.
\end{claim}

\begin{claimproof}
For the sake of contradiction, suppose two such runs exist.
  By generalised soundness or up-to-$K$ soundness of $\W$, there is a run $\m \trans{\gamma} \{\final \colon
  k\}$. Since $\m' \geq \m + \{p \colon 1\}$ and $p$ is not resetable, we have $\{\initial
  \colon k\} \trans{\pi'\gamma} \m''$ for some $\m'' \ge \{\final \colon k, p \colon 1\}$. This contradicts general soundness or up-to-$K$ soundness as $\m'' \not \trans{*} \{\final \colon
  k\}$.
\end{claimproof}

\begin{proof}[Proof of \Cref{lemma:transition_path}]
Towards a contradiction, suppose there exists $t \in T^s$ such that $\pre{t} =
  \emptyset$. By definition, there exists $p \in \post{t}$ for some $p \in P^s$ (otherwise $t$ would be an isolated transition in $\W^s$). Let $u\in T$ be the original nonredundant transition from which $t$ is obtained.
  Let $z, z_u, \zeta,$ and $\rho_u$ be as in the proof of \Cref{lemma:reset_marking}.
  We consider these runs of $\W$:
  \begin{enumerate}
  \item $\{i \colon z_u\} \trans{\rho_u} \m_u \trans{u} \m_u'$ (it exists by nonredundancy);

  \item\label{eq1} $\{i \colon z_u + z\} \trans{\rho_u} \m_u + \{i \colon z\}
    \trans{\zeta} \{f \colon z\} + \Res{\m_u}$;

  \item\label{eq2} $\{i \colon z_u + z\} \trans{\rho_uu} \m_u' + \{i \colon z\} \trans{\zeta} \Res{\m_u'} + \{f \colon z\}$.
  \end{enumerate}
  The last two runs are obtained by \Cref{lemma:reset_marking2}. Observe that
  \[
  \Res{\m_u'} \geq \Res{\m_u} + \{p \colon 1\}.
  \]
  Indeed: since $\pre{t} = \emptyset$, we have $\m_u'(r) \ge \m_u(r)$
  for all $r \in P^s$; and, since $p \in \post{t}$, we also have
  $\m_u'(p) \ge \m_u(p)+1$.  By \Cref{claim:badconf}, the
  runs~\eqref{eq1} and~\eqref{eq2} contradict the generalised
  soundness of $\W$ and up-to-$K$ soundness assuming that $K>z_u+z$.

  The proof of the case where there is a transition $t \in T^s$ such
  that $\post{t} = \emptyset$ is essentially the
  same.
\end{proof}

To prove \Cref{proposition:skeleton} \eqref{skeleton1}, it remains to prove that all places in $P^s$ are on a path from $i$ to $f$. We need to recall some standard definitions from Petri net theory (e.g.\ see~\cite{Rei13}).

A \emph{siphon} is a subset of places $S \subseteq P$ such that for
every transition $t$: if $\post{t} \cap S \neq \emptyset$ then
$\pre{t} \cap S \neq \emptyset$. Similarly, a \emph{trap} is a subset
of places $S'$ such that for every transition $t$: if $\pre{t} \cap S'
\neq \emptyset$ then $\post{t} \cap S' \neq \emptyset$. By definition,
it is readily seen that every unmarked siphon remains unmarked and
every marked trap remains marked. More formally:

\begin{lemma}\label{lem:SiphonTrap}
Let $S$ be a siphon and let $S'$ be a trap of a Petri net (without resets). For every $\m \trans{*} \m'$, the following holds:
\begin{enumerate}
 \item if $\m(S) = 0$, then $\m'(S) = 0$;

 \item if $\m(S') > 0$, then $\m'(S') > 0$.
\end{enumerate}
\end{lemma}

The next two lemmas conclude the proof of \Cref{proposition:skeleton} \eqref{skeleton1}.

\begin{lemma}\label{lem:PathFromItoP}
For every $p \in P^s$, there is a path from $\initial$ to $p$ in $\W^s$.
\end{lemma}

\begin{proof}
  Let $p \in P^s$. Let $X \subseteq P^s$ be the set of places from
  which there is a path to $p$ in $\W^s$. Observe that $X$ is a siphon
  in $\W^s$. Indeed, if $\post{t}\cap X\neq \emptyset$ then
  $\pre{t}\subseteq X$ and $\pre{t}\neq \emptyset$ because of
  \Cref{lemma:transition_path}.

  Note that $p \in X$ and, by \Cref{claim:nonredundant_skeleton}, the place $p$ can be marked in $\W^s$.
Initially, only place $\initial$ is marked. Thus, by \Cref{lem:SiphonTrap}, it must be the case that $\initial \in X$ as otherwise $p$ could not be marked.
\end{proof}

\begin{lemma}\label{lem:PathFromPtoF}
For every $p \in P^s$, there is a
  path from $p$ to $\final$ in $\W^s$.
\end{lemma}

\begin{proof}
  Let $p \in P^s$. Let $X \subseteq P^s$ be the set of places to
  which there is a path from $p$ in $\W^s$. Observe that $X$ is a trap in
  $\W^s$. Indeed, if $\pre{t} \cap X \neq \emptyset$, then $\post{t}
  \subseteq X$, and $\post{t} \neq \emptyset$ because of
  \Cref{lemma:transition_path}.

By \Cref{claim:nonredundant_skeleton}, there exist $k \in \N$ and a marking $\m_s \in \N^{P^s}$ such that $\{\initial
  \colon k\} \transS{*} \m_s$, where $\m_s(p) > 0$ and $\m_s \transS{*} \{\final \colon k\}$.
Note that $p \in X$. Since $\m_s \transS{*} \{\final \colon k\}$ by \Cref{lem:SiphonTrap}, we get $\final \in X$.
\end{proof}

\subsection{Skeletons preserve generalised soundness}\label{sec:skeleton2}

In this subsection, we will
prove \Cref{proposition:skeleton} \eqref{skeleton2}, i.e.\ we show
that if $\W$ is generalised sound (or up-to-$K$ sound for some
sufficiently large $k$), then its skeleton $\W^s$ is also generalised
sound. To do so, we will prove a general lemma, which will also be
useful in the next section. The lemma will not need the assumption of
generalised soundness, but only a weaker property. We start by
defining two properties for reset workflow nets.

We say that a reset workflow net $\W$ has a \emph{full reset run} if there exist $z \in \N$ and a run $\zeta$ such that these three conditions all hold:
\begin{itemize}
\item $\{i \colon z\} \trans{\zeta} \{f \colon z\}$ and $\zeta$ contains
each nonredundant transition of $\W$;

\item if $\m$ is a reachable marking of $\W$, then $\{i \colon z\} + \m \trans{\zeta} \{f \colon z\} + \Res{\m}$;

\item for every decomposition $\zeta = \rho t \rho'$, where $\{i \colon z\} \trans{\rho} \m$, we have $\Res{\m} \trans{*} \{\final \colon z\}$.
\end{itemize}
We call $\zeta$ a \emph{full reset run}. Below, we note that having such a run is a weaker condition than generalised soundness.

\begin{corollary}\label{cor:full_reset}
If $\W$ is generalised sound, then it has a full reset run. Furthermore, the assumption of $\W$ being generalised sound can be relaxed to the weaker assumption that $\W$ is up-to-$K$-sound for an Ackermannianly bounded number $K$.
\end{corollary}

\begin{proof}
The first two conditions follow directly from \Cref{lemma:reset_marking,lemma:reset_marking2}. Let $\zeta$ be a run fulfilling these conditions, and consider a decomposition $\zeta = \rho t \rho'$ with $\{i \colon z\} \trans{\rho} \m$. We have $\{\initial \colon 2z\} \trans{\rho\zeta} \Res{\m} + \{\final \colon z\}$. Since $\W$ is generalised sound or up-to-$K$ sound for $K\geq 2z$, we get the third condition, i.e.\ $\Res{\m} \trans{*} \{\final \colon z\}$.
\end{proof}

\begin{lemma}\label{lemma:fullrun_ackermann}
  There is an algorithm that, given a reset workflow net $\W$,
  outputs: either that $\W$ has a full reset run; or that $\W$ is
  not generalised sound. The procedure works in Ackermannian
  time. Moreover, if there exists $\{i \colon z\} \trans{\zeta} \{f
  \colon z\}$ witnessing a full reset run, then there is also one
  where $z$ is at most Ackermannian. If the algorithm outputs that
  $\W$ is not generalised sound then it computes $K \in \Npos$, which is at most Ackermannian, for which $\W$ is not $K$-sound.
\end{lemma}

\begin{proof}
If the initial place is resetable, then the algorithm outputs that $\W$ is not generalised sound by \Cref{claim:noreset_if}. It also gives an Ackermannian bound for $K$ such that $\W$ is not $K$-sound.

Otherwise, we invoke~\Cref{remark:reset_marking} and obtain that from a marking $\{\initial \colon z\}$, there is a run $\{\initial \colon z\}\trans{\delta}\m$
that executes all nonredundant transitions of $\W$ and is of length at most Ackermannian. Because $\delta$ has a bounded length, we can find it in Ackermannian time. Next, we check if it is possible to cover $\{\final \colon z\}$ from $\m$, using the backward coverability algorithm running in Ackermannian time. If not, then $\W$ is not up-to-$z$ sound and not generalised sound. Otherwise, the algorithm produces a run $\m\trans{\delta'}\upclosure{\{\final \colon z\}}$ of Ackermannian length. We have to check whether $\m \trans{\delta'} \{\final \colon z\}$ holds. If not, then $\W$ is not up-to-$z$ sound and not generalised sound. Otherwise, we pick $\zeta \defeq \delta\delta'$.
%
%


The latter satisfies the first two conditions of a full reset run. We need to check whether it satisfies the third condition.

We do this exhaustively: for any decomposition of $\zeta = \rho\rho'$,
we take a marking $\m$ such that $\{\initial \colon z\} \trans{\rho} \m$
and check a stronger property: whether for a run $\Res{\m}\trans{\xi}\upclosure{\{{\final \colon z}\}}$ it is the case that $\Res{\m}\trans{\xi}\{\final \colon z\}$. Note that this is not as trivial as for $\delta'$ since $\Res{\m}$ does not have to be a reachable configuration.

We prove the property by contradiction. If it does not hold, then $\W$ is not generalised sound and not up-to-$K$ sound. Indeed, due to~\Cref{cor:full_reset}, we know that if $\W$ is up-to-$K$ sound, then it must have a full reset run $\{\initial \colon z'\}\trans{\zeta'}\{\final \colon z'\}$. But, then, there would be a run $\{\initial \colon z+z'\}\trans{\rho}\m + \{\initial \colon z'\} \trans{\zeta'} \Res{\m} +\{\final \colon z'\} \trans{\xi} \{\final \colon z + z'\} + \n$ for some nonempty marking $\n$. This contradicts both generalised soundness and up-to-$K$ soundness assuming $K > z + z'$.

To achieve this check, we use the backward coverability algorithm to find $\xi$, and then we execute $\xi$ step by step. This process works in Ackermannian time.
%
%
%
%
\end{proof}

Below we discuss markings both in $\W$ and its skeleton $\W^s$. For convenience, if $\m_s$ is a marking over $\N^{P^s}$, then we also use it as a marking over $\N^P$, where $\m_s(p) = 0$ for $p \not \in P^s$.

\begin{lemma}\label{lem:OnlyMReachable}
  Let $\W$ be a reset workflow net with a full reset run $\{i \colon z\} \trans{\zeta} \{f \colon z\}$. Let $\{\initial \colon l\} \transS{\pi} \m_s$ be a run of the skeleton $\W^s$, where $l \in \N$.
  There exists $k' \in \N$ such that $\{\initial
  \colon l+k'\} \trans{*} \m_s + \{\final \colon k'\}$ holds in $\W$. Moreover, $k' \le 2z |\pi|$.
\end{lemma}

Before we prove \Cref{lem:OnlyMReachable}, we show how it implies \Cref{proposition:skeleton} \eqref{skeleton2}.

\begin{proof}[Proof of \Cref{proposition:skeleton} \eqref{skeleton2}]
Let $\W$ and $\W^s$ be as described in the proposition.
Suppose that $\W^s$ is not generalised sound. Because of~\cite[Theorem 5.1]{BlondinMO22}, there exists an exponentially bounded number $l$ such that $\{\initial \colon l\} \transS{*} \m_s$ in $\W^s$ and from $\m_s$ it is not possible to reach $\{\final \colon l\}$. Observe that if such $l$ and $\m_s$ exist then the shortest run $\pi$ from $\{\initial \colon l\}$ to $\m_s$ is at most Ackermannian (due to bounds on the length of shortest runs in Petri nets). The bound on the length of the shortest run between two configurations in a Petri net is a consequence of the KLM~\cite{DBLP:conf/stoc/Kosaraju82} algorithm combined with the Ackermannian bound on its complexity~\cite{DBLP:conf/lics/LerouxS19}.
By \Cref{lem:OnlyMReachable} and \Cref{cor:full_reset}, there exists $k' \in \N$ bounded by $2z|\pi|$, i.e. Ackermannianly, such that $\{\initial
  \colon l+k'\} \trans{*} \m_s + \{\final \colon k'\}$ in $\W$. Since $\W$ is generalised sound or up-to-$K$ sound for $K>l+k'$, there is a run $\m_s \trans{*} \{\final \colon l\}$ in $\W$. By \Cref{claim:projection}, this yields to the contradiction with the assumption that $\m_s\not\transS{*}\{\final \colon l\}$ in $\W^s$.
\end{proof}

\begin{proof}[Proof of \Cref{lem:OnlyMReachable}]
  We proceed by induction on the length of the run from $\{\initial
  \colon k\}$ to $\m_s$. If the length is $0$, then the claim is trivial
  and $k' = 0$.

  Suppose the induction claim holds for every marking $\m_s$ reachable
  via a run of length at most $i$. Let $t$ be a transition in $\W^s$ and let
  $\pi' = \pi t$ be a run of length $i + 1$ such that $\{\initial
  \colon k\} \transS{\pi} \m_s \trans{t} \m_s'$. By the induction hypothesis,
  there exists $k'' \in \N$ such that $\{\initial \colon
  k'' + k\} \trans{*} \m_s + \{\final \colon
  k''\}$. Thus, it is sufficient to prove that
  there exists $k_t \in \N$ such that $\{\initial \colon k_t\} + \m_s
  \trans{*} \m_s' + \{\final \colon k_t\}$.

If we treat $\m_s$ and $\m_s'$ as markings over $\N^P$, then $\Res{\m_s} = \m_s$ and $\Res{\m_s'} = \m_s'$.

  Let $\zeta$ be a full reset run in $\W$.
  Let $t'$ be the nonredundant transition in $\W$ from which $t$, in the skeleton, is defined.
 Since $\zeta$ has all nonredundant transitions of $\W$, $t'$ also occurs in it and we can decompose $\zeta = \rho t' \rho'$.
 We denote $\{\initial \colon z\} \trans{\rho} \m_t \trans{t'} \m_t'$.
Consider the following run:
\begin{align*}
  \{\initial \colon 2z\} + \m_s
  &\trans{\rho}
  \m_s + \m_t + \{\initial \colon z\} \\
  &\trans{t'}
  \m_s' + \Rest{t'}{\m_t} + \n +\{\initial \colon z\},
\end{align*}
where $\n(p) = \post{t'}(p)$ if $p \not \in P^s$ and $\n(p) = 0$ otherwise. The last transition splits the effect of $t'$ between $\m_s$ and $\m_t$ as follows. The places $p \in P^s$ are updated by changing $\m_s$ to $\m_s'$. The remaining places are updated by changing $\m_t$ to $\Rest{t'}{\m_t} + \n$. Note that $\Res{\m_t} = \Res{\Rest{t'}{\m_t} + \n}$. Thus, we get
\begin{align*}
  \{\initial \colon z\} + \m_s' + \Rest{t'}{\m_t} + \n
\trans{\zeta}
\m_s' +  \Res{\m_t} + \{\final \colon z\}.
\end{align*}
Finally, recall that by definition of full reset runs (third condition) $\Res{\m_t} \trans{*} \{\final \colon z\}$.
Altogether, we get the following as required:
\[
\{\initial \colon 2z\} + \m_s \trans{*} \m_s' + \{\final \colon 2z \}.
\]
The bound on $k'$ follows directly from the proof.
\end{proof}

\subsection{Proofs of \Cref{theorem:between,theorem:upto}}\label{ssec:thm:proofs}

We may now prove \Cref{theorem:between,theorem:upto} simultaneously.
The constant $k'$ in \Cref{theorem:upto} will be defined at the end of the proof. There, we will observe that if $k$ is large enough then the characterization of $\mathcal{P}_k$ in \Cref{theorem:upto} holds.

Fix $k \in \Npos$. We start with the definition of property $\mathcal{P}_k$.

\begin{definition}\label{def:Pk}
Property $\mathcal{P}_k$ is defined as the conjunction of these properties:
\begin{enumerate}
 \item Places $\initial$ and $\final$ are not resetable;
 
 \item There is a full reset run $\{\initial \colon z\} \trans{\zeta} \{\final \colon z\}$ in $\W$;
 
 \item $\W^s$ is a workflow net which is generalised sound;
 
 \item It is not possible to strictly cover $\{\final \colon k\}$ starting from $\{\initial \colon k\}$, where ``strictly cover'' means reaching some marking $\m > \{\final \colon k\}$. We call this property \emph{coverability-clean};
 
 \item The last property is more complex.
 Consider the following set of markings:
 \begin{align*}
   X \defeq \set{\m \in \N^P \setminus \{\vec{0}\} : \m \not \trans{*} \upclosure{\{\final \colon 1\}}}.
\end{align*}
In words, these are the nonzero markings that cannot mark place $\final$. Let $F \defeq \{\{f \colon \ell\} : \ell \in \N\}$ and $X^{\final} \defeq X + F$, i.e., markings of $X$ with arbitrarily many tokens added to $\final$.
Let $\Res{X^{\final}} \defeq \{\Res{\m} : \m \in X^{\final}\}$.
The last property requires that
 $\{\initial \colon j\} \not\transS{*}
\Res{X^{\final}}\setminus F$ holds for all $j \ge 1$.
\end{enumerate}
\end{definition}

First observe that properties (1--4) are implied by generalised
soundness, as well as up-to-$k$ soundness for sufficiently large
$k$. It is clear for~(4), while properties~(1--3) follow respectively
from \Cref{claim:noreset_if}, \Cref{cor:full_reset} and
\Cref{proposition:skeleton}.

To prove that $\mathcal{P}_k$ is $k$-in-between, it suffices to show
that for any workflow satisfying properties~(1--4), these two claims,
capturing property~(5), hold:

\begin{claim}\label{claim:GeneralOK}
If $\{\initial \colon j\} \transS{*} \Res{X^{\final}}\setminus F$ holds for some $j \ge 1$, then $\W$ is not generalised sound.
\end{claim}

\begin{claim}\label{claim:KsoundOK}
If $\{i \colon j\} \not \transS{*} \Res{X^{\final}} \setminus F$ holds for all $j \in [1..k]$, then $\W$ is up-to-$k$ sound.
\end{claim}


Before proceeding, we show the claim below, which will be helpful for proving \Cref{claim:GeneralOK} and \Cref{claim:KsoundOK}. Let us assume that properties (1--4) hold.

\begin{claim}\label{claim:X}
$\W$ is not $k$-sound if and only if $\{\initial \colon k\} \trans{*} X^{\final}$. 
\end{claim}

\begin{claimproof}
$\Rightarrow$) Suppose $\W$ is not $k$-sound and let $\{\initial \colon k\} \trans{*} \m$ be such that $\m \not \trans{*} \{\final \colon k\}$. Let $\ell$ be the largest number such that from $\m$ we can cover $\{\final \colon \ell\}$. Note that such a number exists and $\ell \le k$, as otherwise we get a contradiction with \Cref{claim:projection} and generalised soundness of $\W^s$. Let $\m \trans{*} \m' + \{\final \colon \ell\}$, where $\m'(\final) = 0$. If $\m' \neq \vec{0}$, then we are done as $\m'\in X$. Otherwise, note that $\ell < k$ as $\m \not\trans{*} \{\final \colon k\}$, which yields a contradiction with \Cref{claim:projection} and generalised soundness of $\W^s$.

$\Leftarrow$) This follows by definition of $X$ and by the fact that, in reset workflow nets, the effect of firing transitions cannot be zero.
\end{claimproof}


\medskip
Let $X_{\vec{0}} \defeq X \cup \{\vec{0}\}$. We may now prove \Cref{claim:GeneralOK}:
\medskip

\begin{claimproof}
Let $\{i \colon j\} \transS{*} \m + \{\final \colon \ell\} \in \Res{X^{\final}} \setminus F$, where $\m(\final) = 0$. Note that $\m \neq \vec{0}$ as the set $F$ was excluded.
By \Cref{lem:OnlyMReachable}, there exists $k'$ such that $\{i \colon j + k'\} \trans{*} \m + \{\final \colon \ell+ k'\}$. As $X_{\vec{0}}$ is downward closed and $\m \neq \vec{0}$, we get $\vec{m} + \{\final \colon \ell + k'\} \in X^{\final}$.
Thus, by~\Cref{claim:X}, $\W$ is not $(j+k')$-sound, and hence it is not generalised sound.
\end{claimproof}

\medskip
We may now prove \Cref{claim:KsoundOK}:
\medskip

\begin{claimproof}
  We show the contrapositive.
  For the sake of contradiction, suppose $\W$ is not up-to-$k$ sound, i.e.\ not $j$-sound for some $j \in [1..k]$. We must exhibit a run $\{\initial \colon j\} \transS{*} \Res{X^{\final}} \setminus F$.

  By \Cref{claim:X}, we have $\{\initial \colon j\} \trans{*} \m + \{\final \colon \ell\} \in X^{\final}$ in $\W$, where $\m(\final) = 0$. By \Cref{claim:projection}, we have $\{\initial \colon j\} \transS{*} \Res{\m} + \{\final \colon \ell\}$ in $\W^s$. If $\Res{\m} \neq \vec{0}$, then we get a contradiction since \[\Res{\m} + \{\final \colon \ell\} \in \Res{X^{\final}} \setminus F.\]
  Suppose $\Res{\m} = \vec{0}$. From this, we have $\{\initial \colon j\}\transS{} \{\final \colon l\}$ which is possible only if $\ell = j$, as otherwise we get a contradiction with $\W^s$ being generalised sound. This means that $\m > \{\final \colon j\}$, which is a contradiction with $\W$ being coverability-clean.
\end{claimproof}
\medskip

To conclude the proof, it suffices to show that $\mathcal{P}_k$ is
 decidable, i.e.\ that properties (1--5) can be checked. First,
 computing $\W^s$ amounts to identifying the subset of nonredundant
 transitions of $\W$, which can be done in Ackermannian time
 by \Cref{proposition:redundancy}. Then:
\begin{enumerate}
 \item Property~(1) is trivial;

 \item Property~(2) is decidable because of~\Cref{lemma:fullrun_ackermann};

 \item Testing whether $\W^s$ is generalised sound can be done in
   PSPACE~\cite[Theorem 5.1]{BlondinMO22};

 \item Property~(4) is a coverability check, which can be done using
   the backward coverability algorithm~\cite{LazicS21} (or
   by~\Cref{lem:boundOnCoverability});
   
 \item Property~(5) requires more effort. We explain it below.
\end{enumerate}

Let us show that we can determine whether $\{\initial \colon k\} \trans{*} X^{\final}$. We start by computing a representation of $X^{\final}$. A representation of $X_{\vec{0}}$ can be computed with the backwards coverability algorithm~\cite{LazicS21}. More precisely, we can compute the set $X'$ of all markings from which there is a run covering $\{\final : 1\}$. Then, $X_{\vec{0}}$ is the complement of $X'$. Moreover, since $X'$ is upward closed, the set $X_{\vec{0}}$ is downward closed. Obviously, this yields a representation of both $X$ and $X^{\final}$.

To simplify the notation, we also identify $\Res{X^{\final}}$ with the set of markings over $P^s$ (i.e.\ by dropping the places outside of $P^s$).
Recall $F = \{\{f \colon \ell\} : \ell \in \N\}$ from the definition of $X^{\final}$.

It remains to show that we can decide whether there exists $j \ge 1$ such that:
\begin{align}\label{eq:j}
  \{i \colon j\} \transS{*} \Res{X^{\final}} \setminus F
  \text{ holds in $\W^s$}.\tag{$\star$}
\end{align}
Moreover, we must prove that if such a $j$ exists, then it is Ackermannianly bounded. This will allow to prove \Cref{theorem:upto}.


We reduce query~\eqref{eq:j} to a reachability query for Petri nets (without resets). To do so, we modify $\W^s$ into a new Petri net $\PN^s$ (whose arcs may consume or produce several tokens at once). First, we add a transition $t_{\initial}$ that can always add one token in place $\initial$, this allows to produce arbitrarily many tokens in $\initial$. Second, we add a place $p_{\text{all}}$ that keeps the sum of tokens in all places from $P^s \setminus \{\final\}$ (it will be needed as $X$ forbids $\vec{0}$). This can be easily achieved by adjusting all transitions on $p_{\text{all}}$ as follows:
\begin{align*}
  \pre{t}(p_{\text{all}})
  \defeq \sum_{\mathclap{p \in P^s \setminus \{\final\}}} \pre{t}(p)
  && \text{ and } &&
  \post{t}(p_{\text{all}})
  \defeq \sum_{\mathclap{p \in P^s \setminus \{\final\}}} \post{t}(p).
\end{align*}
We define $X^{\final}_{\text{all}}$ as the set
\begin{multline*}
  \big\{\m' \in \N^{P^s \cup \{p_{\text{all}}\}} : \exists \m \in X
  \text{ s.t. } \m'(p) = \m(p) \text{ for all } p \in P^s
  \\
  \text{and } \m'(p_{\text{all}}) \ge \m(P^s\setminus \{\final\})\big\}.
\end{multline*}
Note that if $\m' \in X^{\final}_{\text{all}}$, then we have
$\m'(p_{\text{all}}) > 0$. Note that only markings such that
$\m'(p_{\text{all}}) = \m(P^s \setminus \{\final\})$ make sense, but
it will be convenient to allow markings to be larger in place
$p_{\text{all}}$.

Observe that query~\eqref{eq:j} is equivalent to testing whether $\vec{0} \trans{*} X^{\final}_{\text{all}}$ in $\PN^s$. Indeed, transition $t_{\initial}$ allows to guess the initial value $j$, and place $p_{\text{all}}$ guarantees that we at least one token among places other than $\final$.

Now we analyse the set $X^{\final}_{\text{all}}$. Let $\m \in X^{\final}_{\text{all}}$. The following holds:
\begin{itemize}
 \item if $\m' > \m$ and $\m'(p) = \m(p)$ for $p \in P^s \setminus \{\final\}$ then $\m' \in X^{\final}_{\text{all}}$;

 \item if $\m' < \m$ and $\m'(p) = \m(p)$ for $p \in \set{\final, p_{\text{all}}}$ then $\m' \in X^{\final}_{\text{all}}$.
\end{itemize}
Intuitively, $X^{\final}_{\text{all}}$ is downward closed on some places and upward closed on other places.
Since $\PN^s$ is a Petri net (without resets), it is folklore that reachability queries to such sets can be performed in Ackermannian time (see e.g.\
\cite[Lemma 7]{CzerwinskiH22}). Moreover, if there is such a run, then there is one of length at most Ackermannian. This concludes the proof of \Cref{theorem:between}. It also provides an Ackermannian bound on the minimal $j$ satisfying query~\eqref{eq:j}.

We briefly explain that it also proves \Cref{theorem:upto}. Indeed, let us comment on the threshold $k'$ such that, for any $k \ge k'$, $\mathcal{P}_k$ is equivalent to up-to-$k$-soundness.
Observe that properties~(1), (2), (3) and~(5) do not depend on $k$, so intuitively there is a $k'$ such that if they are satisfied for $k'$, then they are satisfied for all $k > k'$. So, properties~(1), (2), (3) and~(5) are implied by
up-to-$k$ soundness for $k > k'$. Moreover, property~(4) is also implied by up-to-$k$ soundness, which means that
$\mathcal{P}_k$, for $k > k'$, is implied by up-to-$k$ soundness.

What remains is to show that an Ackermannianly bounded $k'$ suffices.
Property~(1) is implied by up-to-$k'$ soundness for an Ackermannianly bounded $k'$ according to~\Cref{claim:noreset_if}. Similarly, property~(2) is implied by up-to-$k'$ soundness for an Ackermannianly bounded $k'$ according to~\Cref{lemma:reset_marking}. Property~(3) is implied by up-to-$k'$ soundness for an Ackermannianly bounded $k'$ according to~\Cref{proposition:skeleton}.
Thus, it remains to bound the number $k'$ needed for property~(5).
We know that, if there is a run that violates property~(5), then there is one of length $\ell$ which is at most Ackermannian. Now, because of~\Cref{lem:OnlyMReachable}, we conclude that there is a run
$\{\initial \colon \ell + \ell \cdot 2z \} \trans{*} X^{\final}$, where $z$ is Ackermannianly bounded as in~\Cref{lem:OnlyMReachable}. This, together with~\Cref{claim:X}, shows that $k'> \ell + \ell \cdot 2z$ suffices.
Altogether, an Ackermannianly bounded $k'$ suffices for the proof of~\Cref{theorem:upto}.

\begin{remark}
One may think that the proof of \Cref{theorem:between} is contradictory with the undecidability of generalised soundness, as it might seem that, using \Cref{claim:GeneralOK} and \Cref{claim:KsoundOK}, we can decide generalised soundness. The reason why there is no contradiction is that, earlier, we assumed that $\W$ is coverability-clean. In some sense, checking the coverability-clean property, for all $k$, is the source of undecidability for generalised soundness.
\end{remark}

\section{Conclusion}
In this paper, we studied soundness in reset workflow nets: the
standard correctness notion of a well-established formalism for the
modeling of process activities such as business processes.

All existing variants of soundness, but generalised soundness, were
known to be undecidable for reset workflow nets. In this work, we
have shown that generalised soundness is also undecidable. This closes
its status which had been open for over fifteen years.

Given the resulting undecidable landscape, we investigated a new
approach. We introduced the notion of $k$-in-between soundness, which
lies between $k$-soundness and generalised soundness. We revealed an
unusual complexity behaviour: a decidable soundness property is in
between two undecidable ones. We think this can be valuable in the
algorithmic analysis of reset workflow nets, and that it may spark a
new line of research both in theory and practice.

\subsection{Other future work}

The reachability problem for Minsky machines is already undecidable
for two transitions that test counters for zero. Thus, our proof of
the undecidability of generalised soundness only requires four
transitions that reset some places. The question about decidability of
generalised soundness for reset workflow nets with fewer than four
transitions that reset some places, remains open. We conjecture
decidability for reset workflow net with only one such transition.

Furthermore, it would be interesting to extend the definition of
soundness to more powerful models like well-structured transition systems
(WSTS): the properties of resilience~\cite{DBLP:conf/vmcai/FinkelH24}
can be seen as a first step. We may also try to adapt the efficient
reductions for Petri nets \cite{DBLP:conf/vmcai/AmatDB24} to reset
Petri nets and reset workflow nets.

\balance
\bibliographystyle{ACM-Reference-Format}
\bibliography{references}

\end{document}